\documentclass[11pt]{article}
\usepackage[a4paper,margin=27mm]{geometry}
\usepackage{amsmath,amssymb,amsthm,graphicx,booktabs,placeins,microtype}
\usepackage{xurl}
\usepackage[colorlinks=true,linkcolor=blue,citecolor=blue,urlcolor=blue]{hyperref}
\newtheorem{theorem}{Theorem}[section]
\newtheorem{proposition}[theorem]{Proposition}
\newtheorem{lemma}[theorem]{Lemma}

\newcommand{\Reg}{\operatorname{Reg}}
\graphicspath{{figures/}}
\title{Hyperelliptic sigma solutions of polynomial Hamiltonian systems}
\author{Graham Hesketh}
\date{October 2026}
\begin{document}
\maketitle
\begin{abstract}
We study polynomial Hamiltonian systems whose Hamiltonians consist of two weighted products of canonical modes minus a polynomial in their powers. This class extends the structure of familiar genus-one coupled-mode systems from nonlinear optics. We construct local analytic solutions in terms of sigma functions and Abelian integrals on smooth odd-degree hyperelliptic curves, subject to stated regularity conditions. We associate \^Onishi's determinant identities with the conservation-law structure of these systems, extending the connection previously identified between Frobenius--Stickelberger determinants and conserved quantities in genus-one nonlinear optics. Explicit genus-two and genus-three examples include bounded oscillatory scalar motion. We also classify the Hamiltonian polynomials that eliminate the additional mode exponentials and show that their elementary modes solve second-order equations with self-generating polynomial potentials. We evaluate the solutions using Generapy, a purpose-built Python package for algebraic curves and Abelian functions, and validate them against independent fourth-order Runge--Kutta integration.
\end{abstract}

\section{Introduction}

Many algebraic coupled-mode systems in nonlinear optics admit genus-one elliptic solutions. Examples include four-wave mixing, polarization dynamics and coherent couplers in cubic media, and two-mode and three-mode mixing in quadratic media. These systems are characterized by a conserved polynomial Hamiltonian and Manley--Rowe-like intermodal power laws. Together with their elliptic reductions, these conservation laws reveal a common mathematical structure~\cite{agrawal,armstrong,jensen,stolen,chen-snyder,chen,envelopes,fourwave,coupler}.

The same viewpoint includes the polynomial differential systems defining Jacobi, lemniscatic and Dixon elliptic functions~\cite{dixon}. The coupled-mode systems also admit projections onto the unforced, undamped Duffing oscillator, the Euler top and the Lagrange top. These connections link optical models to the classical and geometric mechanics of integrable systems; background accounts are given in~\cite{holm,holm-rigid,whittaker-dynamics,mechanics,duffing}.

Optical solutions commonly express modal powers through Jacobi or Weierstrass elliptic functions and recover phases through elliptic integrals. In~\cite{envelopes,fourwave,coupler}, we expressed the full complex amplitudes as shifted Weierstrass sigma ratios with exponential factors. In the genus-one setting, we also associated the Frobenius--Stickelberger determinant identities~\cite{frobenius-stickelberger} with the intermodal power laws and the conserved Hamiltonian~\cite{fourwave,coupler}.

The mode count and interaction multiplicities are integer parameters; in optics, the interaction degree is related to the material nonlinearity. Increasing these parameters can produce scalar curves of genus greater than one. What then replaces the elliptic sigma solutions and their determinant conservation laws?

\^Onishi's identities~\cite{onishi} generalize the Frobenius--Stickelberger determinant to hyperelliptic curves of arbitrary genus. They provide a natural setting in which to ask whether the association between sigma solutions, determinant identities and conservation laws extends to higher-genus canonical systems.

The theory of hyperelliptic and, more broadly, Abelian functions has its roots in the nineteenth-century work of Abel, Jacobi, G\"opel, Rosenhain, Weierstrass, Riemann, Klein and Poincar\'e, among others. Baker's books at the turn of the twentieth century brought together and substantially enhanced much of this theory~\cite{baker-abelian,baker-periodic}. A renewed connection with nonlinear integrable equations emerged in the 1970s through finite-gap integration and algebro-geometric solutions, in work by Novikov, Dubrovin, Its, Matveev, Krichever and others~\cite{krichever,finite-gap}. For historical background and introductions to sigma functions, we refer the reader to Baker's books, the review by Buchstaber, Enolski and Leykin~\cite{bel-sigma}, Matsutani's monograph~\cite{matsutani-book} and the references therein.

Buchstaber, Enolskii and Leykin~\cite{bel-kleinian} give a modern treatment of hyperelliptic Kleinian functions, which generalise the Weierstrass elliptic functions to higher genus, including algebraic Jacobi inversion and applications to integrable equations. For the scalar motion considered here, the relevant inversion takes place on the one-point Abel image rather than at a general point of the Jacobian. Matsutani and Previato~\cite{mp-strata,mp-strata-ii} study inversion and sigma-derivative identities on these lower-dimensional strata for cyclic curves, including the odd-degree hyperelliptic family. Enolski and collaborators~\cite{enolski,thirdkind} develop sigma-based inversion methods and apply them to relativistic particle motion, including black-hole space-times. Their third-kind integral formulas, together with the hyperelliptic prime-form representation of Gibbons, Matsutani and \^Onishi~\cite{primeform}, provide tools relevant to individual-mode reconstruction.

The complete mode products also connect with algebraic interpolation and addition laws. Bernatska~\cite{bernatska-fields} constructs curve functions with prescribed zeros through normalized interpolation determinants, while Bernatska and Kopeliovich~\cite{divisor-addition} use hyperelliptic interpolation to describe divisor addition and reduction. These constructions clarify the relation between the determinants, their complementary zeros and the polynomial interpolation used herein. Beyond arguments restricted to individual curve points, Eilbeck, Enolskii and Previato~\cite{eep-addition} derive a genus-two Frobenius--Stickelberger addition formula for three general Jacobian arguments.

\subsection{Our contribution}\label{sec:intro-contribution}

We consider $N$ independent canonical pairs of complex-valued functions $u_j(z),v_j(z)$ of an evolution variable $z$, with Hamiltonian
\[
H(u,v)=\prod_{j=1}^{N}u_j^{\beta_j}
       +\prod_{j=1}^{N}v_j^{\beta_j}
       -Q(I_1,\ldots,I_N),
\]
where $I_j(z)=u_j(z)v_j(z)$ are generally complex quantities, the weights $\beta_j$ are positive integers, and $Q$ is a polynomial in the $I_j$. We call these paired products modal powers, without assuming that they are real or positive. Manley--Rowe-like conservation laws relate their evolution through conserved weighted differences, stated in \eqref{eq:weighted-invariants}. Together with the Hamiltonian, these laws reduce the power dynamics to a single scalar equation on a hyperelliptic curve. Recovering the individual modes requires further integration, since their paired products alone do not determine them.

This reduction relates our systems to the resonant Hamiltonians studied by Martin Kummer~\cite{kummer,holm} and their subsequent treatment through classical and quantum Kummer shape algebras by Odzijewicz and Wawreniuk~\cite{shapes}.\footnote{Here Kummer refers to Martin Kummer, distinct from Ernst Eduard Kummer of algebraic number theory and algebraic geometry.} The conserved weighted differences arise from symmetries of the canonical modes, while the products entering the Hamiltonian describe their collective interaction.

We construct local solutions for the individual complex modes using hyperelliptic sigma functions and Abelian integrals. To our knowledge, this is the first identification of \^Onishi's hyperelliptic determinant identities with the Manley--Rowe-like conservation laws and Hamiltonian structure of polynomial canonical mode systems, extending our earlier genus-one results in nonlinear optics. The two-point identity yields the paired-power relations underlying these conservation laws, while the multipoint identities represent the complete weighted products entering the conserved Hamiltonian. We identify the individual mode factors within these determinants and explain the cancellation of their exponential corrections. Genus-two and genus-three examples illustrate the construction, including bounded oscillatory scalar motion in genus three, and are checked against independent integration of the full canonical equations.

We also classify all polynomials $Q$ for which the additional correction vanishes for every invariant tuple. The nonconstant generators are specific Chan--Chyan--Srivastava multivariable Lagrange polynomials. Their mode pairs have constant Wronskians and solve second-order equations with polynomial potentials; polynomiality fixes the squared Wronskian under a prescribed-product constraint.

For numerical evaluation we use Generapy~\cite{generapy}, a purpose-built Python library for algebraic curves and Abelian functions, which supplies periods, Abelian integrals and sigma functions. Generapy is built on mpmath~\cite{mpmath} for arbitrary-precision numerical arithmetic. During its development, we tested Generapy by reproducing several examples from the literature, including the period and Kleinian-function computations for Bernatska's genus-three trigonal curve~\cite{bernatska-computation}; the reproduced examples are described in the Generapy documentation~\cite{generapy}.

Sections~\ref{sec:setup}--\ref{sec:modes} develop scalar reduction, the sigma representation of the scalar coordinate and mode reconstruction; Section~\ref{sec:onishi-products} establishes the determinant connection, and Section~\ref{sec:polynomial-potentials} derives the second-order polynomial potentials. Sections~\ref{sec:examples} and~\ref{sec:numerical} present the examples and comparisons with independent fourth-order Runge--Kutta integration. The solutions are local and apply to smooth odd-degree hyperelliptic curves, subject to the regularity conditions stated in the main results.

\section{Scalar reduction and separated mode equations}\label{sec:setup}

We start from the polynomial Hamiltonian and derive the common scalar motion of the modal powers. The same calculation expresses each mode equation in terms of that scalar trajectory and separates its logarithmic derivative into a fixed-point contribution and a polynomial correction. These are the two terms integrated in the following sections.

\subsection{Hamiltonian equations and conserved powers}

Let $u_j(z),v_j(z)\in\mathbb C$, $j=1,\ldots,N$, be independent complex canonical modes, with primes denoting derivatives with respect to $z$. Define

\[
I_j=u_jv_j,\qquad
U=\prod_{j=1}^N u_j^{\beta_j},\qquad
V=\prod_{j=1}^N v_j^{\beta_j},
\]

where the weights $\beta_j$ are positive integers. We consider

\begin{equation}\label{eq:hamiltonian}
H(u,v)=U+V-Q(I_1,\ldots,I_N),
\end{equation}

with $Q$ a polynomial in the powers. Here ``power'' denotes $I_j=u_jv_j$ and need not be real or positive. Write $Q_j=\partial Q/\partial I_j$, treating the powers as independent arguments of $Q$ before evaluating them along a trajectory. Hamilton's equations are

\begin{equation}\label{eq:canonical}
\begin{aligned}
u_j'&=\frac{\partial H}{\partial v_j}
=\beta_jv_j^{\beta_j-1}\prod_{\substack{1\le k\le N\\k\ne j}}v_k^{\beta_k}-Q_j(I)u_j,\\
v_j'&=-\frac{\partial H}{\partial u_j}
=-\beta_ju_j^{\beta_j-1}\prod_{\substack{1\le k\le N\\k\ne j}}u_k^{\beta_k}+Q_j(I)v_j.
\end{aligned}
\end{equation}

The Hamiltonian is conserved because it has no explicit dependence on $z$. Multiplying the first equation by $v_j$ and the second by $u_j$ cancels the terms containing $Q_j$ and gives

\[
I_j'=\beta_j(V-U).
\]

Thus every normalized power $I_j/\beta_j$ has the same derivative. In particular, the $N-1$ independent weighted differences

\begin{equation}\label{eq:weighted-invariants}
C_j=\beta_N I_j-\beta_j I_N,\qquad j=1,\ldots,N-1,
\end{equation}

are conserved. These are the Manley--Rowe-like intermodal power laws of the canonical system.

\subsection{A common scalar coordinate}

The conserved differences allow all powers to be written using a single scalar coordinate. Fix the additive freedom in the scalar coordinate by defining

\[
\rho(z)=-\frac{\sum_{j=1}^{N} I_j(z)}{\sum_{j=1}^{N}\beta_j},\qquad
\gamma_j=I_j+\beta_j\rho.
\]

The power equations imply

\[
\rho'=U-V,\qquad \gamma_j'=0,\qquad \sum_{j=1}^{N}\gamma_j=0.
\]

Consequently

\begin{equation}\label{eq:power-restriction}
I_j=\gamma_j-\beta_j\rho,
\end{equation}

and the conserved differences take the form

\[
\beta_k I_j-\beta_j I_k
=\beta_k\gamma_j-\beta_j\gamma_k.
\]

The constants $\gamma_j$ are therefore fixed by the initial powers, while $\rho$ describes their common evolution. Equivalently, $-\rho$ is the weighted mean of the normalized powers $I_j/\beta_j$, with weights $\beta_j$. 

On this power line, define the polynomial restrictions

\begin{equation}\label{eq:qr-defs}
\begin{aligned}
q(\rho)&=Q(\gamma_1-\beta_1\rho,\ldots,\gamma_N-\beta_N\rho),\\
r(\rho)&=\prod_{j=1}^N(\gamma_j-\beta_j\rho)^{\beta_j}=UV.
\end{aligned}
\end{equation}

The constant term of $Q$ sets the zero of energy without changing the equations of motion. We choose it so that the trajectory under consideration lies on $H=0$, giving $U+V=q(\rho)$. Combining the sum and difference of $U,V$ gives

\begin{equation}\label{eq:products}
U=\frac{q(\rho)+\rho'}2,\qquad
V=\frac{q(\rho)-\rho'}2,
\end{equation}

and therefore

\begin{equation}\label{eq:scalar}
(\rho')^2=(U-V)^2=(U+V)^2-4UV
=q(\rho)^2-4r(\rho)=F(\rho).
\end{equation}

This is the scalar curve governing the modal powers. Here and below, $q_\rho$ and $F_\rho$ denote polynomial derivatives with respect to $\rho$. Along a nonconstant trajectory the scalar equation also gives

\[
\rho''=\frac12F_\rho(\rho).
\]

It also follows directly by differentiating $\rho'=U-V$ with Hamilton's equations, so it remains valid at turning points, where $\rho'=0$. It determines continuation there without choosing a square-root sign anew.

\subsection{Separated equations for the individual modes}

Knowing the powers does not yet determine the individual modes. On a chart where $I_j\ne0$, divide the canonical equations by $u_j,v_j$ and substitute the product formulas \eqref{eq:products}:

\begin{equation}\label{eq:mode-log}
\begin{aligned}
\frac{u_j'}{u_j}
&=\frac{\beta_jV}{I_j}-Q_j(I)
=\frac{q(\rho)-\rho'}{2[\gamma_j/\beta_j-\rho]}-Q_j(I(\rho)),\\
\frac{v_j'}{v_j}
&=-\frac{\beta_jU}{I_j}+Q_j(I)
=\frac{-q(\rho)-\rho'}{2[\gamma_j/\beta_j-\rho]}+Q_j(I(\rho)).
\end{aligned}
\end{equation}

Each right-hand side depends only on the scalar trajectory $\rho,\rho'$. Once that trajectory is known, the modes follow by separate logarithmic integrations, with initial amplitudes inherited from the original canonical system.

The constants $\gamma_j/\beta_j$ and $q(\gamma_j/\beta_j)$ can be expressed as coordinates of fixed points on the scalar curve. Introduce $\mu_j$ by

\begin{equation}\label{eq:mode-points}
\rho(\mu_j)=\frac{\gamma_j}{\beta_j},\qquad
\rho'(\mu_j)=q(\rho(\mu_j)).
\end{equation}

The notation $\mu_j$ labels the fixed point with the displayed coordinates, which lies on the scalar curve because $r(\rho(\mu_j))=0$. Adding and subtracting the fixed polynomial value gives

\begin{equation}\label{eq:split-log}
\begin{aligned}
\frac{u_j'}{u_j}
&=\frac{\rho'(\mu_j)-\rho'}{2[\rho(\mu_j)-\rho]}+B_j(\rho),\\
\frac{v_j'}{v_j}
&=\frac{-\rho'(\mu_j)-\rho'}{2[\rho(\mu_j)-\rho]}-B_j(\rho),
\end{aligned}
\end{equation}

where

\begin{equation}\label{eq:B-def}
B_j(\rho)=
\frac{q(\rho)-q(\rho(\mu_j))}{2[\rho(\mu_j)-\rho]}
-Q_j(I(\rho)).
\end{equation}

The divided difference is polynomial because its numerator vanishes at $\rho=\rho(\mu_j)$; the restricted partial derivative is polynomial as well. Thus $B_j$ is polynomial, determined by the original Hamiltonian. In particular, differentiating $Q$ before restricting it to the power line retains information needed for individual modes that is not contained in $q$ alone.

\subsection{Hyperelliptic scope of this paper}

We express the separated mode integrals using sigma functions and Abelian differentials in the following setting:

\begin{itemize}
\item $F$ is squarefree of odd degree $2g+1$, with $g\ge2$.
\item The fixed coordinates $\rho(\mu_j)$ are distinct.
\end{itemize}

The first condition gives a smooth compactified hyperelliptic curve of genus $g$ with one point at infinity. The sigma formulas are presented in the leading-four convention. For a general nonzero leading coefficient $f$, set $\tau=sz$ and $Y=\rho'/s$, with $s^2=f/4$, so that $Y^2=F/s^2$ has leading coefficient four. This is a reparametrization of propagation for curve evaluation; the original modes and Hamiltonian are retained. The corresponding mode formula is specified after Theorem~\ref{thm:mode-solutions}. The weights $\beta_j$ encode the multiplicities of coincident mode points after a compatible degenerate-mode reduction; distinctness refers to the resulting weighted list, rather than to each occurrence in its divisor. At a branch fixed point $a_j=\rho(\mu_j)$, we have $q(a_j)=0$ and $F'(a_j)=-4r'(a_j)\ne0$, so smoothness forces $\beta_j=1$.

Even-degree models and singular curves require corresponding changes to the curve representation and are not treated here.

\section{Hyperelliptic representation of the scalar coordinate \texorpdfstring{$\rho$}{rho}}\label{sec:scalar}

The scalar reduction supplies the curve, its fixed mode points and the separated mode equations. We first express the scalar coordinate $\rho$ through derivatives of the hyperelliptic sigma function restricted to the one-point Abel image; Section~\ref{sec:modes} then reconstructs the individual modes. We specify the curve, Abel coordinates and sigma conventions, state the coordinate representation and its low-genus forms, and finally describe their evaluation at prescribed propagation.

\subsection{The scalar curve}

Let the smooth compactification of
\begin{equation}\label{eq:curve}
y^2=F(\rho)=4\rho^{2g+1}+\cdots
\end{equation}
have genus $g\ge2$, with its unique point at infinity. The variable $y$ denotes the curve ordinate; along a nonconstant scalar trajectory, $y=\rho'(z)$. Choose an initial finite curve point $(\rho(z_0),\rho'(z_0))$ satisfying the scalar equation~\eqref{eq:scalar}. The moving point determines a one-point Abel image on which the coordinate can be recovered by sigma derivatives.

\subsection{Abel coordinates and propagation}

The holomorphic differentials and the Abel lift of the moving point are
\begin{equation}\label{eq:holomorphic-diffs}
\omega_k=\frac{\rho^{k-1}\,d\rho}{y},\qquad k=1,\ldots,g,
\end{equation}
\begin{equation}\label{eq:abel-map}
w_k(z)=\int_\infty^{(\rho(z),\rho'(z))}\omega_k.
\end{equation}
Choose paths from infinity to the initial moving point and the fixed mode points, and continue the moving path with the trajectory. The Abel vectors are lifts, rather than vectors reduced modulo periods; subsequent sigma expressions and second-kind integrals use compatible path choices. Along the trajectory,
\begin{equation}\label{eq:w-derivatives}
w_1(z)=w_1(z_0)+z-z_0,\qquad w_k'(z)=\rho(z)^{k-1}.
\end{equation}
The initial offset is retained because the Abel map is based at infinity, whereas $z_0$ is an arbitrary propagation origin. Only $w_1$ is linear in propagation. For genus one this sole coordinate is enough to evaluate the familiar elliptic solution directly. For $g\ge2$, the remaining coordinates are constrained by the requirement that the full Abel vector represent a single curve point. They cannot be replaced by independently chosen linear functions of $z$.

\subsection{Sigma and second-kind normalization}

Write $F(\rho)=\sum_{m=0}^{2g+1}\lambda_m\rho^m$, where $\lambda_{2g+1}=4$, and put $\lambda_m=0$ outside this range. We use the algebraic second-kind basis \cite[Eq.~II.2]{thirdkind}
\begin{equation}\label{eq:second-kind-diffs}
dr_k=\frac{1}{4y}\sum_{m=k}^{2g+1-k}
(m+1-k)\lambda_{m+1+k}\rho^m\,d\rho,\qquad k=1,\ldots,g.
\end{equation}
Together with \eqref{eq:holomorphic-diffs}, this fixes the Kleinian sigma exponential convention. More explicitly, let $2\omega,2\omega'$ be the first-kind period matrices and $2\eta=-\oint_a dr$ the second-kind $a$-period matrix in the same homology basis. Then
\begin{equation}\label{eq:sigma-definition}
\sigma(w)=C\exp\!\left(\tfrac12w^T\kappa w\right)
\theta[\delta]\bigl((2\omega)^{-1}w;\tau\bigr),\quad
\kappa=\eta\omega^{-1},\quad \tau=\omega^{-1}\omega',
\end{equation}
where $\theta[\delta]$ is the Riemann theta function with $\delta$ chosen as the characteristic of the Riemann constant based at infinity. The constant $C$ fixes the leading Schur--Weierstrass polynomial. We use \^Onishi's sigma normalization \cite{onishi}, as implemented in generapy:
\begin{align}
\sigma(w)&=w_1-\tfrac13w_2^3+\text{higher-weight terms},&&g=2,\label{eq:sigma-leading-g2}\\
\sigma(w)&=w_1w_3-w_2^2-\tfrac13w_2w_3^3+\tfrac1{45}w_3^6
+\text{higher-weight terms},&&g=3.\label{eq:sigma-leading-g3}
\end{align}
The polynomial for genus three also agrees with \cite[Eq.~3.39]{enolski}. References to multipoint determinant identities are understood with \^Onishi's published errata. The two-point addition identity used below retains the sign stated in \cite[Lemma~7.1]{onishi}.

Let $W_n$ denote the Abel image of effective degree-$n$ divisors, based at infinity. For $g\ge2$, $W_1\subset W_{g-1}$, and $W_{g-1}$ is the sigma zero divisor. Consequently $\sigma(w(z))=0$. In genus three the one-point image also satisfies $\sigma_3=\sigma_{33}=0$ \cite[Section~VI]{thirdkind}.

Sigma subscripts denote partial derivatives in the Abel coordinates. The derivatives used for the mode kernel are \cite{onishi}
\begin{equation}\label{eq:sharp-flat}
\sigma_\sharp=\left(\prod_{\substack{2\le k\le g\\k\ {\rm even}}}\partial_{w_k}\right)\sigma,
\qquad
\sigma_\flat=\left(\prod_{\substack{3\le k\le g\\k\ {\rm odd}}}\partial_{w_k}\right)\sigma.
\end{equation}
An empty product takes no derivative: $(\sigma_\sharp,\sigma_\flat)=(\sigma_2,\sigma)$ for genus two and $(\sigma_2,\sigma_3)$ for genus three. These kernel factors are distinct from the derivative ratios used to recover $\rho$.

\subsection{From Abel inversion to a sigma coordinate ratio for \texorpdfstring{$\rho(z)$}{rho(z)}}

We now express the coordinate of the moving curve point through sigma functions evaluated at its Abel vector. This is the one-point restriction of Jacobi inversion: the arguments lie on $W_1$, rather than being arbitrary points of the Jacobian. The resulting formula is an analytic representation on that constrained image; evaluation at a given $z$ also requires finding the corresponding Abel vector, as discussed in Section~\ref{sec:scalar-evaluation}.

Ordinary Jacobi inversion starts with the summed Abel vector of $g$ curve points and recovers their coordinates. Our trajectory contains only one moving point, so we use a limiting restriction of that problem. Sending $g-1$ divisor points to infinity leaves a single finite point. Ordinary sigma quotients can then have vanishing numerators and denominators, so one must expand in local coordinates at infinity and retain the first nonzero derivatives. Hyperelliptic addition and determinant identities control these limits. The resulting derivative indices are determined by the gap sequence at infinity.

For the differential ordering \eqref{eq:holomorphic-diffs}, the general result is particularly concise. Define
\begin{equation}\label{eq:inversion-operator}
\mathcal D_g=\prod_{\substack{4\le k\le g\\k\ {\rm even}}}\partial_{w_k}.
\end{equation}
The product is empty in genera two and three. Since $\sigma_\sharp=\partial_{w_2}\mathcal D_g\sigma$, the scalar inversion theorem on the one-point stratum gives
\begin{proposition}[One-point sigma coordinate recovery]\label{prop:inversion}
For $g\ge2$ and $w(z)$ on the one-point Abel image of a finite curve point,
\begin{equation}\label{eq:inversion-general}
\rho(z)=-\frac{\partial_{w_1}\mathcal D_g\sigma(w(z))}
{\partial_{w_2}\mathcal D_g\sigma(w(z))}.
\end{equation}
\end{proposition}
\begin{proof}
In the notation of \cite[Definition~5.1 and Theorem~5.5]{primeform}, the one-point derivative index is $\natural_1=\{2,4,6,\ldots\}\cap\{1,\ldots,g\}$. Replacing its index $2$ by $1$ gives $\natural_1^{(1)}$. The $k=1$ stratum inversion theorem states that the ratio of these derivatives is minus the coordinate of the single finite point. These two derivatives are precisely $\partial_{w_1}\mathcal D_g\sigma$ and $\partial_{w_2}\mathcal D_g\sigma$. The denominator is $\sigma_\sharp$, which is nonzero at every finite point of this image by \cite[Proposition~6.6(1)]{onishi}. An overall multiplicative normalization of sigma cancels from this ratio. Replacing the monic ordinate by $y/2$ gives the same Abel differential basis after the translation to our curve convention.
\end{proof}

The formula represents the scalar coordinate analytically on its constrained Abel image. Its origin can also be seen directly in the two-point addition identity \cite[Lemma~7.1]{onishi}. If $w=w(z)$ and $w_\infty(t)$ is the Abel vector of a point with $\rho=t^{-2}$ tending to infinity, then
\begin{equation}\label{eq:inversion-infinity-limit}
\begin{aligned}
&(-1)^{g+1}\frac{\sigma_\flat(w+w_\infty(t))\sigma_\flat(w-w_\infty(t))}
{\sigma_\sharp(w)^2\sigma_\sharp(w_\infty(t))^2}
=t^{-2}-\rho(z),\\
&w_{\infty,k}(t)=-\frac{t^{2g-2k+1}}{2g-2k+1}
+O(t^{2g-2k+3}).
\end{aligned}
\end{equation}
Subtracting the known $t^{-2}$ pole and extracting the finite part recovers $-\rho(z)$. Expanding the sigma functions makes this a finite-derivative procedure; the stratified inversion theorem identifies the shorter ratio \eqref{eq:inversion-general}. Sending only one auxiliary point to infinity in the two-point identity gives this finite-part construction, while the full Jacobi-divisor restriction sends $g-1$ points to infinity.

\subsection{Explicit coordinate ratios in genera two, three and four}

The general ratio becomes $-\sigma_1/\sigma_2$ for both genera two and three, and $-\sigma_{14}/\sigma_{24}$ for genus four. The following equivalent low-genus representations are useful in the conventions of Enolski and collaborators \cite{enolski,thirdkind}:
\begin{align}
\rho(z)&=-\frac{\sigma_1(w(z))}{\sigma_2(w(z))},&&g=2,\label{eq:inversion-g2}\\
\rho(z)&=-\frac{\sigma_{13}(w(z))}{\sigma_{23}(w(z))},&&g=3,\label{eq:inversion-g3}\\
\rho(z)&=-\frac{\sigma_{144}(w(z))}{\sigma_{244}(w(z))},&&g=4.\label{eq:inversion-g4}
\end{align}
All are evaluated on the one-point image, with nonzero displayed denominators. In genus two it is characterized locally by $\sigma=0$; in genus three the relevant branch satisfies $\sigma=\sigma_3=0$; in genus four it satisfies $\sigma=\sigma_4=\sigma_{44}=0$. The genus-four formula is quoted directly from \cite[Eq.~VII.18]{thirdkind}. These higher-derivative formulas are alternative representations on the same stratum, rather than the definition of the arbitrary-genus ratio.

For genus two, differentiating $\sigma(w(z))=0$ gives $\sigma_1+\rho\sigma_2=0$. In genus three $\sigma_3=0$ removes the last term when differentiating $\sigma=0$, giving the same first-derivative ratio; differentiating $\sigma_3=0$ and using $\sigma_{33}=0$ gives the displayed alternative \eqref{eq:inversion-g3}. Such identities explain why different derivative orders can represent the same coordinate.

\subsection{Evaluation at prescribed propagation}\label{sec:scalar-evaluation}

The sigma ratios above give exact coordinate identities on $W_1$. For $g\ge2$, however, only $w_1$ is linear in propagation. Computing $\rho(z)$ directly from these formulas requires determining the remaining Abel coordinates while enforcing the one-point-image constraints and compatible path lifts. This is a constrained inversion problem in addition to the evaluation of sigma functions.

Our numerical implementation uses the scalar quadrature
\begin{equation}\label{eq:scalar-inverse-integral}
z-z_0=\int_{(\rho(z_0),\rho'(z_0))}^{(\rho(z),\rho'(z))}\frac{d\rho}{y}.
\end{equation}
The integration path and sheet are continued from the initial curve point. We invert this relation numerically to obtain the moving point, calculate its Abel lift using~\eqref{eq:abel-map}, and evaluate the individual-mode sigma formulas with Generapy~\cite{generapy}. We also evaluate the scalar sigma ratios at the computed Abel vectors as consistency checks. These checks verify the coordinate identities on the recovered image; they do not provide an independent sigma-based inversion at prescribed propagation. Thus the formulas are analytic, while the reported evaluations use numerical quadrature, inversion and special-function computation.

\section{Individual mode solutions}\label{sec:modes}

Once the scalar trajectory is known, the remaining task is to recover the individual modes rather than just their products. We integrate the fixed-point contribution by sigma functions, determine its exponential correction, and then combine it with the polynomial correction to obtain the full canonical solution.

\subsection{Fixed mode points and the mode differential}

The fixed mode points were defined in \eqref{eq:mode-points}. The numerical trajectory need not visit these fixed points. The fixed lift $w(\mu_j)$ is calculated directly from the curve point, using paths compatible with \eqref{eq:abel-map}.

The mode differential is
\begin{equation}\label{eq:third-kind-diff}
\omega_j^{\rm mode}
=\frac{\rho'(\mu_j)-y}{2[\rho(\mu_j)-\rho]}\frac{d\rho}{y}.
\end{equation}
It has residue $+1$ at $(\rho(\mu_j),-\rho'(\mu_j))$ and residue $-1$ at infinity. Along the moving point it is exactly the rational term of $u_j'/u_j$ in \eqref{eq:split-log} times $dz$. The following subsections integrate this differential through the sharp/flat sigma factors and combine it with the polynomial correction $B_j$.

\subsection{The fixed-point second-kind coefficients}

In arbitrary genus, define
\begin{equation}\label{eq:h-general}
h_{jk}=-\Reg\int_\infty^{(\rho(\mu_j),\rho'(\mu_j))}dr_k.
\end{equation}
Here $\Reg$ subtracts the divergent powers of the local coordinate $t=\rho^{-1/2}$ at infinity and chooses zero constant term in the primitive there. The sheet and the lifted path are the same as for $w(\mu_j)$. These choices fix the integration constants and signs; $h_{jk}$ depend on the fixed point and not on propagation.

In genera two and three, put $\zeta_{\sharp,k}=\partial_{w_k}\log\sigma_\sharp$. The coefficients have the short forms
\begin{equation}\label{eq:h-formulas}
\begin{aligned}
h_{j2}&=\tfrac12\zeta_{\sharp,2}(w(\mu_j)),\\
h_{j1}&=\zeta_{\sharp,1}(w(\mu_j))
-\frac{\sigma_1(w(\mu_j))}{\sigma_2(w(\mu_j))}h_{j2},\\
h_{j3}&=\zeta_{\sharp,3}(w(\mu_j))\qquad (g=3).
\end{aligned}
\end{equation}
In genus two $\sigma_1/\sigma_2=-\rho(\mu_j)$, so $h_{j1}=\zeta_{\sharp,1}+\rho(\mu_j)h_{j2}$. The genus-three expressions are the negatives of the second-kind primitives in \cite[Proposition~III.2, Eq.~III.10]{thirdkind}; the finite-base constants there are replaced by the zero finite part at infinity in \eqref{eq:h-general}. The same infinity prescription gives the first two expressions in genus two. These short expressions are convenient low-genus evaluations; the regularized-integral definition applies in every genus, as does the finite sigma-based coefficient extraction described after Lemma~\ref{lem:kernel}.

\subsection{The sigma logarithmic differential}

\begin{lemma}\label{lem:kernel}
With the first- and second-kind bases and sigma normalization above, define on compatible regular lifts the uncorrected point factors (denoted by a bar)
\[
\bar u_j(z)=\frac{\sigma_\flat(w(z)+w(\mu_j))}
{\sigma_\sharp(w(z))\sigma_\sharp(w(\mu_j))},\qquad
\bar v_j(z)=\frac{\sigma_\flat(w(z)-w(\mu_j))}
{\sigma_\sharp(w(z))\sigma_\sharp(w(\mu_j))}.
\]
Then
\begin{equation}\label{eq:kernel-log}
\frac{d}{dz}\log\bar u_j(z)
=\frac{\rho'(\mu_j)-\rho'(z)}{2[\rho(\mu_j)-\rho(z)]}
+\sum_{k=1}^g h_{jk}\rho(z)^{k-1}.
\end{equation}
For $\bar v_j$ replace $\rho'(\mu_j)$ and $h_{jk}$ by their negatives. Moreover,
\begin{equation}\label{eq:kernel-product}
\bar u_j(z)\bar v_j(z)=(-1)^{g+1}[\rho(\mu_j)-\rho(z)].
\end{equation}
\end{lemma}

\begin{proof}
We describe the normalization translation behind the third-kind formula. On the lifted curve, the sharp/flat kernel is the ratio of the algebraic prime form between the moving point and the sheet involution of the fixed point to its infinity factors. This is the hyperelliptic prime-form representation of \cite[Theorem~7.17]{primeform}. To translate its monic curve convention to ours, use the ordinate $y/2$; the holomorphic basis $d\rho/(2(y/2))$ is then \eqref{eq:holomorphic-diffs}. The second-kind basis is fixed by \eqref{eq:second-kind-diffs}, rather than by an independent holomorphic shift.

The logarithmic kernel differential and \eqref{eq:third-kind-diff} have the same simple poles and residues; their difference is a holomorphic differential in the moving point. The sigma quasi-period multipliers have logarithmic derivatives given by periods of $dr_k$, so compatible continuation preserves this difference in the moving lift. Temporarily write the difference as $\sum_{k=1}^{g} K_k(P)\omega_k$, where $P=\mu_j$ is the fixed curve point.

To identify its coefficients, vary the fixed curve point, holding the moving point fixed. The mixed logarithmic derivative of the algebraic prime form is the algebraic fundamental bidifferential. In \cite[Eq.~II.11]{thirdkind}, that bidifferential is the derivative of the rational Cauchy differential plus $dr^T\omega$. Evaluating the fixed argument on its opposite sheet changes $dr$ to $-dr$. Subtracting the derivative of the rational mode differential therefore gives
\[
d_P\bigl(d_z\log\bar u_j(z)-\omega_j^{\rm mode}\bigr)
=-\sum_{k=1}^{g} dr_k(P)\,\omega_k(z),
\qquad d_P K_k(P)=-dr_k(P).
\]
Here $d_P$ denotes variation along the fixed-point copy of the curve, not a propagation derivative. Thus $K_k(P)=-R_k(P)+C_k$, where $R_k=\Reg\int_\infty^P dr_k$ and $C_k$ is constant in $P$.

To fix $C_k$, use \^Onishi's two-point addition formula \cite[Lemma~7.1]{onishi}, which gives \eqref{eq:kernel-product} independently of the logarithmic formula. Its prefactor is $(-1)^{g+1}$ in our normalization. Let $\iota P$ have the opposite ordinate and the negative Abel lift. The corresponding kernel is $\bar v_j$ up to a factor independent of the moving point. The logarithmic derivative of the product is $-\rho'/(\rho(\mu_j)-\rho)$, exactly the sum of the two rational mode terms. Hence $K_k(\iota P)=-K_k(P)$. Near infinity, $\iota$ sends $t=\rho^{-1/2}$ to $-t$; each $dr_k$ has only even powers of $t$ multiplying $dt$. Its zero-constant-term primitive is therefore odd, so $R_k(\iota P)=-R_k(P)$. Substitution into $K_k=-R_k+C_k$ gives $C_k=-C_k$, hence $C_k=0$. Compatible continuation preserves this identity on the lifted curve. Consequently $K_k(P)=h_{jk}$ as defined in~\eqref{eq:h-general}. Substituting $\omega_k=\rho^{k-1}dz$ proves \eqref{eq:kernel-log}, and sheet reversal gives the formula for $\bar v_j$. This is the infinity-based third-kind integration formula; Eq.~III.17 of \cite{thirdkind} instead uses ordinary sigma with finite-base shifts.

These identities also hold at a finite branch fixed point $(a,0)$.
By \cite[Proposition~6.6(1)]{onishi}, $\sigma_\sharp$ has no zeros on
the one-point Abel image except at infinity. In a local curve coordinate
$t$ with $\rho-a=t^2$, write $y=t s(t^2)$, where
$s(\xi)=\sqrt{F(a+\xi)/\xi}$ is a locally chosen holomorphic square
root with $s(0)^2=F'(a)\ne0$. Both kinds of Abelian differential
are therefore holomorphic. Hence the fixed Abel lift, the
coefficients~\eqref{eq:h-general}, and the sigma quotients extend
holomorphically as the fixed point approaches $(a,0)$. Away from the
moving zeros, \eqref{eq:kernel-product} ensures nonzero factors, so their
logarithmic derivatives extend as well. Passing to the limit proves
\eqref{eq:kernel-log}--\eqref{eq:kernel-product} without changing the
shifted sigma ratios or their exponential factors.
\end{proof}

Alternatively, in any genus, the coefficients $h_{jk}$ can be recovered directly from the sigma kernel. Subtracting the rational term in \eqref{eq:kernel-log} from its left-hand side leaves a polynomial in $\rho$ of degree at most $g-1$. Its coefficients are determined by evaluations at $g$ regular curve points with distinct $\rho$-coordinates, or by its first $g$ Taylor coefficients in $\rho$ at one regular nonbranch point. The points must avoid poles and zero sigma denominators, and all evaluations use compatible Abel lifts. This is coefficient extraction from an exact polynomial identity, rather than a fit to the mode trajectory; the low-genus formulas above give shorter evaluations of the same constants.

\subsection{Corrected point factors as regularized restrictions of Baker--Akhiezer functions}\label{sec:corrected-point-factors}

Put $\Delta w_k(z)=w_k(z)-w_k(z_0)$. The corrected point factors are
\begin{equation}\label{eq:point-factor-sigma}
\begin{aligned}
\widehat u_j(z)&=\sqrt{\beta_j}\,\bar u_j(z)
\exp\!\left[-\sum_{k=1}^g h_{jk}\Delta w_k(z)\right],\\
\widehat v_j(z)&=(-1)^{g+1}\sqrt{\beta_j}\,\bar v_j(z)
\exp\!\left[\sum_{k=1}^g h_{jk}\Delta w_k(z)\right].
\end{aligned}
\end{equation}
The corrected point factors $\widehat u_j,\widehat v_j$
in~\eqref{eq:point-factor-sigma} arise as regularized restrictions of
Baker--Akhiezer functions to the one-point Abel image. Matsutani's sigma
representation~\cite[Section~4.2.2.3, Eq.~(4.42)]{matsutani-book} uses the
same $\sigma_\sharp$ normalization, while \^Onishi's stratum
limits~\cite[Proposition~6.6]{onishi} yield the
$\sigma_\flat/\sigma_\sharp$ quotient appearing here. When $B_j=0$,
these factors, normalized by their initial values, give the complete modes.

Lemma~\ref{lem:kernel} gives
\begin{equation}\label{eq:point-factors}
\begin{aligned}
\frac{\widehat u_j'}{\widehat u_j}
&=\frac{\rho'(\mu_j)-\rho'(z)}{2[\rho(\mu_j)-\rho(z)]},\\
\frac{\widehat v_j'}{\widehat v_j}
&=\frac{-\rho'(\mu_j)-\rho'(z)}{2[\rho(\mu_j)-\rho(z)]},
\end{aligned}
\end{equation}
and
\begin{equation}\label{eq:point-product}
\widehat u_j(z)\widehat v_j(z)=\beta_j[\rho(\mu_j)-\rho(z)]=I_j(z).
\end{equation}
Consequently, with continued logarithms,
\begin{equation}\label{eq:third-kind-integral}
\int_{(\rho(z_0),\rho'(z_0))}^{(\rho(z),\rho'(z))}\omega_j^{\rm mode}
=\log\frac{\widehat u_j(z)}{\widehat u_j(z_0)}.
\end{equation}
\subsection{The mode-solution theorem}\label{sec:mode-solutions}

The polynomial correction in each mode is integrated by the unique decomposition
\begin{equation}\label{eq:B-reduction}
\frac{B_j(\rho)}y\,d\rho=d\bigl(P_j(\rho)y\bigr)
+\sum_{k=1}^g c_{jk}\omega_k+\sum_{k=1}^g d_{jk}\,dr_k,
\end{equation}
where $P_j$ is polynomial and the coefficients are constant on the chosen power line. To construct it, observe that
\begin{equation}\label{eq:exact-polynomial-differential}
d(\rho^\ell y)=\left(\ell\rho^{\ell-1}F+\tfrac12\rho^\ell F_\rho\right)\frac{d\rho}{y},\qquad \ell\ge0.
\end{equation}
Its numerator has degree $\ell+2g$ and leading coefficient $4\ell+4g+2$. If the remaining numerator has leading term $a\rho^h$, $h\ge2g$, take $\ell=h-2g$ and subtract $a/(4\ell+4g+2)$ times this exact differential. The leading term cancels, lowering the degree. Repetition leaves degree at most $2g-1$. The numerators of $dr_k$ in \eqref{eq:second-kind-diffs} have successive degrees $2g-k$ and leading coefficients $2g+1-2k$, $k=1,\ldots,g$; subtracting their multiples cancels degrees $2g-1,\ldots,g$. The remainder is a linear combination of $\omega_1,\ldots,\omega_g$. These distinct leading degrees also prove uniqueness.

Along the scalar trajectory, \eqref{eq:exact-polynomial-differential} is the identity
$d(\rho^\ell\rho')/dz=\ell\rho^{\ell-1}F+\rho^\ell F_\rho/2$.
Thus the relation $(\rho')^2=F(\rho)$ expresses higher powers of $\rho$ as derivatives of $\rho^\ell\rho'$ plus lower powers; integration records the derivative terms explicitly and reduces the remainder to a finite differential basis.

Let $R_k$ be a primitive of $dr_k$, continued along the moving path, and write $\Delta R_k(z)=R_k(z)-R_k(z_0)$. Since $\rho'=y$, integration gives
\begin{equation}\label{eq:B-integral}
\mathcal C_j(z):=\int_{z_0}^z B_j(\rho(s))\,ds
=\bigl[P_j(\rho)\rho'\bigr]_{z_0}^z
+\sum_{k=1}^g c_{jk}\Delta w_k(z)+\sum_{k=1}^g d_{jk}\Delta R_k(z).
\end{equation}
The increments use compatible lifts and may be evaluated directly from the finite initial point. Polynomial differentials $B_j\,d\rho/y$ have only a possible pole at infinity and zero residue there. Their primitives therefore consist of the algebraic, first-kind and second-kind terms displayed above.

When $\deg B_j\le g-1$, uniqueness gives $P_j=0$ and $d_{jk}=0$ for every $k=1,\ldots,g$, so
\begin{equation}\label{eq:degree-condition}
B_j(\rho)=\sum_{k=1}^g c_{jk}\rho^{k-1},\qquad
\mathcal C_j(z)=\sum_{k=1}^g c_{jk}\Delta w_k(z).
\end{equation}
Only first-kind differentials are then needed. This occurs automatically for $\deg Q\le g$ and can also occur through cancellation for higher-degree $Q$.

\begin{theorem}[Local hyperelliptic mode solutions]\label{thm:mode-solutions}
Let $\beta_j$ be positive integers, $\sum_{j=1}^{N}\gamma_j=0$, and $Q$ a polynomial in independent powers. Define $q,r,F$ by \eqref{eq:qr-defs}--\eqref{eq:scalar}. Assume that $F$ is squarefree of degree $2g+1$, $g\ge2$, and that the fixed coordinates $\gamma_j/\beta_j$ are distinct. The displayed formulas use leading coefficient $4$; for a general nonzero leading coefficient, use the prescription~\eqref{eq:general-leading-normalization} below. Let $[\rho'(z)]^2=F(\rho(z))$, $\rho''(z)=F_\rho(\rho(z))/2$ define a nonconstant scalar trajectory on a simply connected local chart with $I_j\ne0$ and nonvanishing displayed sigma normalizers and point factors. Use compatible infinity-based lifts and \eqref{eq:h-general}.

Choose initial modes with
\[
u_j(z_0)v_j(z_0)=\gamma_j-\beta_j\rho(z_0),\qquad
\prod_{j=1}^{N} u_j(z_0)^{\beta_j}=\frac{q(\rho(z_0))+\rho'(z_0)}2.
\]
Then the local canonical solutions are
\begin{equation}\label{eq:mode-solutions}
\begin{aligned}
u_j(z)&=u_j(z_0)\frac{\widehat u_j(z)}{\widehat u_j(z_0)}
\exp\!\left[\mathcal C_j(z)\right],\\
v_j(z)&=v_j(z_0)\frac{\widehat v_j(z)}{\widehat v_j(z_0)}
\exp\!\left[-\mathcal C_j(z)\right].
\end{aligned}
\end{equation}
They satisfy $u_jv_j=I_j$, the product identities \eqref{eq:products}, $H=0$, and the weighted invariants \eqref{eq:weighted-invariants}.
\end{theorem}

\begin{proof}
The point-product identity and initial paired powers give $u_jv_j=I_j$. Differentiating \eqref{eq:mode-solutions}, then using \eqref{eq:point-factors} and \eqref{eq:B-integral}, gives
\[
\frac{u_j'}{u_j}
=\frac{\rho'(\mu_j)-\rho'(z)}{2[\rho(\mu_j)-\rho(z)]}+B_j(\rho(z))
=\frac{\beta_j[q(\rho(z))-\rho'(z)]}{2I_j(z)}-Q_j(I(z)).
\]
This is the $u_j$ equation in \eqref{eq:mode-log}; the $v_j$ equation follows in the same way. To verify the complete products from the prescribed scalar data, the chain rule gives
\[
q_\rho=-\sum_{j=1}^{N}\beta_jQ_j(I),\qquad
r_\rho=-r\sum_{j=1}^{N}\frac{\beta_j^2}{I_j}.
\]
Together with $\rho''=qq_\rho-2r_\rho$ and $q^2-(\rho')^2=4r$, these imply
\[
\frac{d}{dz}\log\frac{q+\rho'}2
=\frac{q-\rho'}2\sum_{j=1}^{N}\frac{\beta_j^2}{I_j}
-\sum_{j=1}^{N}\beta_jQ_j(I)
=\sum_{j=1}^{N}\beta_j\frac{u_j'}{u_j}.
\]
The initial product condition therefore gives $U=(q+\rho')/2$. Since $UV=r$, it also gives $V=(q-\rho')/2$. Substituting these products into the logarithmic equations recovers the polynomial canonical equations \eqref{eq:canonical}. Finally, $H=0$ follows from $U+V=q$, and the weighted invariants follow from $I_j=\gamma_j-\beta_j\rho$. The canonical vector field is polynomial, so local uniqueness for analytic ordinary differential equations identifies the constructed modes with the solution of the prescribed initial-value problem on this chart.
\end{proof}

For a general leading coefficient $f\ne0$, the same mode formula uses the normalized curve
\begin{equation}\label{eq:general-leading-normalization}
 \tau=sz,\qquad Y=\rho'/s,\qquad s^2=f/4,\qquad Y^2=F(\rho)/s^2.
\end{equation}
Compute the point factors on this curve, with fixed ordinates $b_j/s$, and take $\mathcal C_j=\int_{\tau_0}^{\tau} (B_j/s)\,d\tau$ using its differential basis. The modes are still evaluated at $\tau=sz$ with their original initial values. Indeed,
\[
 \frac{d}{d\tau}\log u_j
 =\frac{b_j/s-Y}{2(a_j-\rho)}+\frac{B_j}{s},
 \qquad a_j=\gamma_j/\beta_j,\quad b_j=q(a_j),
\]
so multiplication by $s$ recovers~\eqref{eq:split-log}; the $v_j$ equation follows likewise. The paired powers and complete products therefore retain their original Hamiltonian values. This extends the reconstruction to any nonzero leading coefficient without rescaling the modes.

The nonzero-power and nonvanishing-factor conditions specify a chart for the logarithmic proof, rather than singularities of the polynomial canonical equations. At a modal zero, those equations can continue a finite solution even though its logarithmic derivative is undefined; the sigma expressions must then be interpreted by continuation or limits. Finite turning points likewise do not obstruct scalar continuation. Continuation around curve cycles is a separate issue and can introduce multipliers, as discussed in Section~\ref{sec:discussion}.

\subsection{Polynomials with vanishing mode correction}\label{sec:special-Q}

Recall from~\eqref{eq:B-def} that, with $a_j=\gamma_j/\beta_j$,
\[
 B_j(\rho)=\frac{q(\rho)-q(a_j)}{2(a_j-\rho)}
             -Q_j(I(\rho)),\qquad
 q(\rho)=Q(I(\rho)),\quad I_i(\rho)=\gamma_i-\beta_i\rho.
\]
Here $Q_j=\partial Q/\partial I_j$ is taken before restriction to the power
line. We call $Q$ universally special if every $B_j$ vanishes identically
for every invariant tuple
$\boldsymbol\gamma$ satisfying $\sum_{j=1}^{N}\gamma_j=0$.
Then $\mathcal C_j=0$ in~\eqref{eq:mode-solutions}, so the complete modes
are the normalized Baker--Akhiezer factors of
Section~\ref{sec:corrected-point-factors}.

The following classification is algebraic and does not require the
odd-degree curve assumptions.

We use the Chan--Chyan--Srivastava multivariable Lagrange
polynomials~\cite{ccs,ccs-umbral} with parameters
$-\boldsymbol\beta/2=(-\beta_1/2,\ldots,-\beta_N/2)$, defined by
the ordinary generating function
\begin{equation}\label{eq:special-ccs}
 \prod_{j=1}^{N}(1-x_jt)^{\beta_j/2}
 =\sum_{n=0}^{\infty}g_n^{(-\boldsymbol\beta/2)}(x)t^n.
\end{equation}
(These multivariable polynomials are distinct from the Weierstrass
invariants $g_2,g_3$.)

\begin{theorem}[Universally special polynomials]\label{thm:special-Q}
Let $M=\sum_{j=1}^{N}\beta_j$. If $M$ is odd, the universally special
polynomials are precisely the constants. If $M=2m$, they are precisely
\begin{equation}\label{eq:special-Q}
 Q(I)=\kappa g_m^{(-\boldsymbol\beta/2)}
                  (I_1/\beta_1,\ldots,I_N/\beta_N)+E,
\end{equation}
where $\kappa,E$ are arbitrary constants.
No symmetry assumption on $Q$ is required.
\end{theorem}

\begin{proof}
Put $x_j=I_j/\beta_j$ and
$T(x)=Q(\beta_1x_1,\ldots,\beta_Nx_N)$.
At a point of the power line, $x_j=\gamma_j/\beta_j-\rho$, so
vanishing of~\eqref{eq:B-def} is equivalent to
\begin{equation}\label{eq:special-identity}
 T(x)-T(x-x_j\mathbf1)
 =\frac{2x_j}{\beta_j}\partial_jT(x),\qquad j=1,\ldots,N,
\end{equation}
where $\mathbf1=(1,\ldots,1)$. This identity holds for arbitrary $x$:
choose $\rho=-\sum_{i=1}^{N}\beta_i x_i/M$ and
$\gamma_i=\beta_i(x_i+\rho)$. Conversely it implies $B_j=0$, including
$x_j=0$ by polynomial continuation.

Each homogeneous part of $T$ satisfies~\eqref{eq:special-identity}
separately, since the substitutions are linear and $x_j\partial_j$
preserves degree. Consider a nonzero homogeneous solution of degree $d>0$.
We first show that $T(\mathbf1)\ne0$. Adding the same scalar $s$ to every
variable gives a polynomial in $s$,
\[
 T(x+s\mathbf1)=\sum_{k=0}^r A_k(x)s^k,\qquad A_r\ne0,\quad r\le d.
\]
Apply~\eqref{eq:special-identity} at $x+s\mathbf1$ and clear its denominator:
\[
 2(x_j+s)\partial_jT(x+s\mathbf1)
 =\beta_j\{T(x+s\mathbf1)-T(x-x_j\mathbf1)\}.
\]
The coefficient of $s^{r+1}$ gives $\partial_jA_r=0$ for every $j$.
Thus $A_r$ is a nonzero constant. It is homogeneous of degree $d-r$,
so $r=d$ and $A_d=T(\mathbf1)\ne0$.

At $x=\mathbf1$, equation~\eqref{eq:special-identity} gives
$2\partial_jT(\mathbf1)=\beta_jT(\mathbf1)$, since $T(0)=0$.
Summing and using Euler's homogeneous identity yields
$2dT(\mathbf1)=MT(\mathbf1)$, hence $2d=M$.
Two solutions of this degree are proportional: subtract a multiple to
make their difference vanish at $\mathbf1$, then apply the preceding
argument. This proves necessity and uniqueness up to scale and constants.

For existence when $M=2m$, let $\Phi_x(t)$ denote the generating
function in~\eqref{eq:special-ccs}. Differentiation gives
\[
 \partial_jg_m^{(-\boldsymbol\beta/2)}(x)=-\frac{\beta_j}{2}
             \sum_{k=0}^{m-1}x_j^{m-1-k}g_k^{(-\boldsymbol\beta/2)}(x).
\]
Also
\[
 \begin{aligned}
 \Phi_{x-x_j\mathbf1}(t)
 &=(1+x_jt)^m\Phi_x\!\left(\frac{t}{1+x_jt}\right),\\
 g_m^{(-\boldsymbol\beta/2)}(x-x_j\mathbf1)
 &=\sum_{k=0}^m x_j^{m-k}g_k^{(-\boldsymbol\beta/2)}(x).
 \end{aligned}
\]
The last identity follows by taking the coefficient of $t^m$.
Together these formulas give~\eqref{eq:special-identity} for
$g_m^{(-\boldsymbol\beta/2)}(x)$.
Finally $g_m^{(-\boldsymbol\beta/2)}(s\mathbf1)=(-s)^m$,
so the solution is nonzero.
\end{proof}

Theorem~\ref{thm:special-Q} selects the degree-$m$ member of this
established polynomial family with parameters $-\boldsymbol\beta/2$.
Here
``Lagrange polynomial'' refers to this generating function, rather than
to the cardinal polynomials of Lagrange interpolation.

The theorem requires the same $Q$ to cancel every $B_j$ for every invariant
tuple. On a single power line, additional polynomials can give the same
dynamics: if $L$ is affine linear and vanishes on that line, adding
$L^2R$ to $Q$ changes neither its value nor its first partial derivatives
there, for any polynomial $R$. Such additions generally change the
dynamics on other power lines.

The parity condition concerns the total weight $M=\sum_{j=1}^{N}\beta_j$.
For unit weights $M=N$; with unequal weights, an odd number of mode pairs
can have even total weight and admit a nonconstant special $Q$.

For constant $Q=E$, our convention $H=0$ gives the energy level $U+V=E$.
The constant does not change Hamilton's equations, but it enters the scalar
curve:
\begin{equation}\label{eq:special-constant-F}
 F(\rho)=E^2-4r(\rho).
\end{equation}
Here $\deg F=\deg r=M$, so odd $M=2g+1$ gives genus $g$ when $F$ is
squarefree. For nonconstant special $Q$, cancellation in $q^2-4r$ can
instead lower the degree; the total-weight-ten example in
Section~\ref{sec:special-examples} has a degree-five, genus-two curve.

\section{\^Onishi determinants for the complete products}\label{sec:onishi-products}

The mode solutions and their Hamiltonian product identities have already been established in Theorem~\ref{thm:mode-solutions}, independently of the determinant representation:
\[
U(z)=\prod_{j=1}^{N} u_j(z)^{\beta_j}=\frac{q(\rho(z))+\rho'(z)}2,\qquad
V(z)=\prod_{j=1}^{N} v_j(z)^{\beta_j}=\frac{q(\rho(z))-\rho'(z)}2.
\]
We now relate these products to \^Onishi's identities for configurations admitting the distinct-point partition described below. A repeated mode point can occur once in each of the main and residual lists, accommodating some weights $\beta_j>1$ without derivative rows. This connects the reconstructed modes to the same addition identities that give their paired powers. The essential step is to account for every uncorrected point factor in the determinant's sigma expression, including those contained in its composite factor; this will also explain why the exponential corrections disappear from the complete products.

We know the number $N$ of mode pairs, the genus $g$ of the scalar curve and the polynomial $q$. We now ask whether the two curve functions on the right can be represented by \^Onishi determinants. The size $n$ of the required matrix has yet to be determined. We choose it from the columns available for this genus and the degree of $q$. We begin with that choice, before relating the determinant's sigma expression to the mode solutions.

\subsection{Step 1: determine the columns and size, then fit the product}

The determinant formulas in this section use the leading-four convention. For a general leading coefficient, apply them to the normalized curve~\eqref{eq:general-leading-normalization} with ordinate $Y$ and polynomial $q/s$; multiplying the resulting curve functions $(q/s\pm Y)/2$ by $s$ recovers the original products $U,V$.

An \^Onishi determinant is formed by evaluating a prescribed list of curve functions at $n$ points, one point per row. The standard column functions are polynomials in the two curve coordinates, ordered by their poles at infinity \cite[Theorem~7.2]{onishi}. They include both pure powers of the coordinate and mixed coordinate--ordinate monomials. The determinant uses the first $n$ functions in this order; we cannot choose an arbitrary collection of columns. Here $\rho$ has pole order two and the monic ordinate $\rho'/2$ has pole order $2g+1$. The beginning of the ordered list is therefore
\begin{equation}\label{eq:onishi-pole-columns}
1,\rho,\ldots,\rho^g,\rho'/2,\rho^{g+1},\rho\rho'/2,\rho^{g+2},\ldots,
\end{equation}
with pole orders $0,2,\ldots,2g,2g+1,2g+2,2g+3,2g+4,\ldots$. The factor $1/2$ converts our leading-four curve to \^Onishi's monic convention: $(\rho'/2)^2=F(\rho)/4$, with the same Abel differential $d\rho/[2(\rho'/2)]=d\rho/\rho'$.

To obtain $q\pm\rho'$, the matrix must contain the ordinate column and enough polynomial columns to express the already known $q$. We choose the shortest initial column list that meets both requirements:
\[
\begin{array}{c|c|c}
\text{known scalar degree}&\text{required columns}&\text{chosen size }n\\\hline
\deg q\le g &(1,\rho,\ldots,\rho^g,\rho'/2)&g+2\\
\deg q=g+1 &(1,\rho,\ldots,\rho^g,\rho'/2,\rho^{g+1})&g+3\\
d=\deg q>g+1 &\text{initial list through }\rho^d&2d-g+1.
\end{array}
\]
The first two column arrangements also occur in the interpolation determinants used by Bernatska and Kopeliovich for hyperelliptic divisor reduction~\cite[Sections~3.1--3.2]{divisor-addition}, up to the ordinate normalization and row order. We first explain the product construction for these two arrangements, including the column operation in the genus-two example, then give the calculation for the longer lists in Section~\ref{sec:longer-columns}.

Even when $\deg q<g$, we need $g+2$ columns to reach the ordinate. For $\deg q=g+1$, we add the next polynomial column. Thus $n$ is determined here by $g$ and the scalar degree, whereas $N$ was determined by the number of canonical pairs in the Hamiltonian. For the two worked examples presented later in Section~\ref{sec:examples}, these choices give
\begin{center}
\begin{tabular}{cccc}\toprule
Known genus $g$ & Known degree $\deg q$ & Known mode count $N$ & Chosen size $n$\\\midrule
2 & 3 & 5 & 5\\
3 & 4 & 7 & 6\\
\bottomrule\end{tabular}
\end{center}

For these first two arrangements there is exactly one ordinate column, $\rho'/2$, and every other column is a power of $\rho$. Expanding in the moving row gives a constant coefficient of $\rho'/2$, because its cofactor contains only fixed rows. The calculation below identifies that nonzero coefficient as a Vandermonde minor. For higher scalar degree the initial list also contains mixed columns such as $\rho\rho'/2$; at the fixed points, substituting $\rho'(\mu_j)=q(\rho(\mu_j))$ makes these entries polynomial in the fixed coordinates.

We can now choose $n-1$ fixed points for the matrix and append the moving point as its last row. The fixed mode points are
\[
P_j=(\rho(\mu_j),\rho'(\mu_j)),\qquad
\rho'(\mu_j)=q(\rho(\mu_j)),
\]
and the moving point is $P(z)=(\rho(z),\rho'(z))$. For the cubic genus-two example, the size just derived is $n=5$. Evaluating these column functions at the fixed points $P_1,P_2,P_3,P_4$ gives the first four rows; evaluation at $P(z)$ gives the last row:
\begin{equation}\label{eq:onishi-g2-matrix}
D_5(z)=\det\begin{pmatrix}
1&\rho(\mu_1)&\rho(\mu_1)^2&\rho'(\mu_1)/2&\rho(\mu_1)^3\\
1&\rho(\mu_2)&\rho(\mu_2)^2&\rho'(\mu_2)/2&\rho(\mu_2)^3\\
1&\rho(\mu_3)&\rho(\mu_3)^2&\rho'(\mu_3)/2&\rho(\mu_3)^3\\
1&\rho(\mu_4)&\rho(\mu_4)^2&\rho'(\mu_4)/2&\rho(\mu_4)^3\\
1&\rho(z)&\rho(z)^2&\rho'(z)/2&\rho(z)^3
\end{pmatrix}.
\end{equation}
Let $A$ denote the indices of the chosen fixed rows; here $A=\{1,2,3,4\}$. Write $q(\rho)=q_0+q_1\rho+q_2\rho^2+q_3\rho^3$, and denote the $j$th matrix column by $C_j$. Subtracting a linear combination of the other columns from a column leaves the determinant unchanged. We therefore subtract half this polynomial from the fourth column:
\[
C_4\longmapsto C_4-\tfrac12(q_0C_1+q_1C_2+q_2C_3+q_3C_5).
\]
Every fixed entry becomes zero because $\rho'(\mu_j)=q(\rho(\mu_j))$. Only the moving entry remains:
\[
C_4\longmapsto
\begin{pmatrix}0\\0\\0\\0\\{}[\rho'(z)-q(\rho(z))]/2\end{pmatrix}.
\]
Expansion in this column has sign $(-1)^{5+4}=-1$ and leaves the fixed-point Vandermonde minor. Hence
\[
D_5(z)=\Delta_A\frac{q(\rho(z))-\rho'(z)}2,\qquad
\Delta_A=\prod_{1\le i<j\le4}[\rho(\mu_j)-\rho(\mu_i)].
\]
We have made the determinant reproduce $V$ directly. Let $\iota$ (the Greek letter iota) denote the sheet involution, which reverses the ordinate while preserving the coordinate: $\iota(P_j)=(\rho(\mu_j),-\rho'(\mu_j))$ and write $P_j^*=\iota(P_j)$. Here the star denotes the sheet involution, not complex conjugation. Reversing the ordinates of the fixed points, that is $\rho'(\mu_j)\mapsto-\rho'(\mu_j)$ for $j\in A$, while leaving the moving point unchanged, replaces their polynomial values by $-q$. The same operation then gives $D_5^*(z)=-\Delta_A[q(\rho(z))+\rho'(z)]/2$, reproducing $U$.

The same column operation works in any genus for $n=g+2$ with $\deg q\le g$, or for $n=g+3$ with $\deg q=g+1$. The polynomial $q$ is already specified; its fixed-point values allow us to subtract it from the ordinate column. Distinct fixed coordinates ensure that the remaining Vandermonde minor is nonzero. Thus the column operation expresses the determinant as a constant multiple of the complete mode product $V=(q-\rho')/2$.

\begin{lemma}\label{lem:onishi-interpolation}
Let $A$ index $n-1$ chosen fixed points, and form $D_n$ by evaluating the stated column functions at these points and at $P(z)$. Let $D_n^*$ instead evaluate its fixed rows at $P_j^*=\iota(P_j)$, leaving the moving point unchanged. Define
\[
\Delta_A=\prod_{\substack{i,j\in A\\i<j}}[\rho(\mu_j)-\rho(\mu_i)],\qquad
\varepsilon=\begin{cases}1,&n=g+2,\\-1,&n=g+3.\end{cases}
\]
Then
\begin{equation}\label{eq:onishi-interpolation}
D_n(z)=\varepsilon\Delta_A\frac{\rho'(z)-q(\rho(z))}{2},\qquad
D_n^*(z)=\varepsilon\Delta_A\frac{\rho'(z)+q(\rho(z))}{2}.
\end{equation}
\end{lemma}
\begin{proof}
Subtract $q/2$ from the ordinate column and expand its sole nonzero entry, as above. The ordinate is the last column when $n=g+2$; it precedes the last polynomial column when $n=g+3$, giving the sign $\varepsilon$. For the fixed points on the opposite sheet subtract $-q/2$ instead.
\end{proof}

\subsection{The Schur complement and the Vandermonde inverse}

The same determinant reduction can be expressed as a Schur complement. This makes the role of interpolation explicit: the fixed ordinates determine the polynomial part of the moving row through the inverse of the fixed-coordinate matrix.

For the two column arrangements just considered, first move the ordinate column to the last position. For $n=g+2$ it is already there; for $n=g+3$ this requires one column interchange. The resulting matrix has the block form
\begin{equation}\label{eq:onishi-schur-block}
\mathcal M_n(z)=\begin{pmatrix}
\mathsf W_A & \boldsymbol{\rho}'_A/2\\
\mathbf p(\rho(z))^T & \rho'(z)/2
\end{pmatrix},\qquad
\det\mathcal M_n(z)=\varepsilon D_n(z),
\end{equation}
where, in the order of the fixed rows,
\[
\mathbf p(\rho)=(1,\rho,\ldots,\rho^{n-2})^T,\qquad
\mathsf W_A=\bigl(\rho(\mu_j)^k\bigr)_{j\in A,\,0\le k\le n-2},\qquad
\boldsymbol{\rho}'_A=\bigl(\rho'(\mu_j)\bigr)_{j\in A}.
\]
Here $\boldsymbol{\rho}'_A$ is a column vector, and $\mathsf W_A$ is the $(n-1)\times(n-1)$ Vandermonde matrix with determinant $\Delta_A$. It is invertible because the selected fixed coordinates are distinct.

For these column arrangements, $\deg q\le n-2$. Writing $\mathbf q=(q_0,\ldots,q_{n-2})^T$ for its coefficient vector, padded with zeros when necessary, the fixed-point values satisfy
\[
\mathsf W_A\mathbf q=\boldsymbol{\rho}'_A.
\]
Consequently,
\begin{equation}\label{eq:onishi-vandermonde-q}
q(\rho)=\mathbf p(\rho)^T\mathsf W_A^{-1}\boldsymbol{\rho}'_A.
\end{equation}
Thus the Vandermonde inverse takes the fixed ordinates $\rho'(\mu_j)$ to the polynomial $q$ evaluated at the moving coordinate. If $\deg q>n-2$, the same expression gives only the polynomial interpolant of degree at most $n-2$ at those nodes and does not recover the entire $q$. This is interpolation of the already specified scalar restriction, expressed in matrix form.

The Schur complement of $\mathsf W_A$ in \eqref{eq:onishi-schur-block} is
\begin{equation}\label{eq:onishi-schur-complement}
\mathcal S(z)=\frac{\rho'(z)}2
 -\mathbf p(\rho(z))^T\mathsf W_A^{-1}\frac{\boldsymbol{\rho}'_A}{2}
 =\frac{\rho'(z)-q(\rho(z))}{2}.
\end{equation}
Block elimination gives $\det\mathcal M_n=\det\mathsf W_A\,\mathcal S$, and hence
\[
\frac{D_n(z)}{\Delta_A}=\varepsilon\mathcal S(z).
\]
For the later worked examples $\varepsilon=-1$, so $D_n/\Delta_A=-\mathcal S=V$. Applying the sheet involution only to the fixed points negates $\boldsymbol{\rho}'_A$ and leaves $\mathsf W_A$ unchanged. The corresponding Schur complement is $[\rho'(z)+q(\rho(z))]/2=U$, so $U=\varepsilon D_n^*/\Delta_A$. In the original $n=g+3$ column order, the upper-left block is not the Vandermonde matrix; the column interchange and its sign are essential to this interpretation.

\subsection{Higher scalar degree and mixed ordinate columns}\label{sec:longer-columns}

For $d=\deg q\ge g+1$, let $q_d$ be the leading coefficient and put $m=d-g-1$. The initial list ending at $\rho^d$ contains
\[
1,\rho,\ldots,\rho^d,\qquad y/2,\rho y/2,\ldots,\rho^m y/2,
\]
interleaved by pole order. There are $d+1$ polynomial columns and $d-g$ ordinate columns, giving $n=2d-g+1$. Choose $n-1$ distinct fixed coordinates on $y=q(\rho)$ and append the moving row as before. Subtracting $q/2$ from the $y/2$ column again makes all its fixed entries zero. Expansion along this column leaves a fixed-row minor. On these fixed rows, the remaining mixed entries satisfy $\rho^k y/2=\rho^k q(\rho)/2$, so the minor has the polynomial columns
\begin{equation}\label{eq:onishi-longer-minor}
1,\rho,\ldots,\rho^d,\rho q/2,\ldots,\rho^m q/2,
\end{equation}
after reordering. Their successive degrees are $0,\ldots,d+m=n-2$. The leading coefficients are one for the first $d+1$ columns and $q_d/2$ for the last $m$; triangular column elimination therefore gives the minor $(q_d/2)^m\Delta_A$.

The $y/2$ column has position $g+2$, so its cofactor sign is $(-1)^{n+g+2}=-1$. Moving the remaining $m$ ordinate columns past the polynomial columns requires $m(m+1)/2$ interchanges, giving the combined sign $(-1)^{1+m(m+1)/2}$. Reversing the fixed sheets replaces $q_d$ by $-q_d$. Thus, defining
\begin{equation}\label{eq:onishi-longer-signs}
\varepsilon_\pm=(-1)^{1+m(m+1)/2}(\pm q_d/2)^m,\qquad m=d-g-1,
\end{equation}
we obtain
\begin{equation}\label{eq:onishi-longer-interpolation}
D_n=\varepsilon_+\Delta_A\frac{y-q(\rho)}2,\qquad
D_n^*=\varepsilon_-\Delta_A\frac{y+q(\rho)}2.
\end{equation}
In particular $\varepsilon_-=(-1)^m\varepsilon_+$. When $d=g+1$, $m=0$ and both constants equal $-1$, agreeing with Lemma~\ref{lem:onishi-interpolation}. For $d\le g$, put $\varepsilon_+=\varepsilon_-=1$ so \eqref{eq:onishi-longer-interpolation} also includes the $n=g+2$ case.

For example, a quartic $q$ on a genus-two curve requires the seven columns
\[
1,\rho,\rho^2,y/2,\rho^3,\rho y/2,\rho^4.
\]
After subtraction and expansion in the fourth column, the fixed minor has columns
\[
1,\rho,\rho^2,\rho^3,\rho q/2,\rho^4.
\]
Interchanging the last two gives successive degrees zero through five and leading-coefficient product $q_4/2$. Together with the cofactor sign $(-1)^{7+4}=-1$, this yields $\varepsilon_+=q_4/2$ and $\varepsilon_-=-q_4/2$. The mixed ordinate column has contributed a constant prefactor; the moving factor is still $(y-q)/2$.

\subsection{Step 2: relate the sigma factors to the modes}

Before counting the mode factors, we display the sigma side of the identity. For the genus-two five-row determinant \eqref{eq:onishi-g2-matrix}, \^Onishi's formula \cite[Theorem~7.2]{onishi}, with its corrected sign, reads
\begin{equation}\label{eq:onishi-g2-full-sigma}
D_5(z)=-\frac{
\displaystyle\sigma\!\left(w(z)+\sum_{i=1}^4w(\mu_i)\right)
\prod_{1\le i<j\le4}\sigma\!\left(w(\mu_i)-w(\mu_j)\right)
\prod_{i=1}^4\sigma\!\left(w(\mu_i)-w(z)\right)}{
\displaystyle\sigma_2(w(z))^5\prod_{i=1}^4\sigma_2(w(\mu_i))^5}.
\end{equation}
The right side contains two kinds of moving sigma arguments: differences with individual fixed points, and a sum involving all four fixed points. The fixed--fixed differences are constants. Grouping the four individual moving factors with their normalizers gives the uncorrected point factors $\bar v_i$ defined in Lemma~\ref{lem:kernel}. What remains, apart from a fixed constant, is
\[
\frac{\sigma(w(z)+w(\mu_1)+\cdots+w(\mu_4))}{\sigma_2(w(z))}.
\]
We call this the \emph{composite factor} because its shift combines the fixed Abel vectors. We will show that it supplies the additional uncorrected point factors required by the complete mode product. The same grouping applies to all the initial column lists above in any genus.

The determinant vanishes when the moving point satisfies $\rho'=q(\rho)$. The identity
\[
(q-\rho')(q+\rho')=4r(\rho)
 =4\prod_{j=1}^{N}[\gamma_j-\beta_j\rho]^{\beta_j}
\]
locates every such zero, including its multiplicity. At a nonbranch fixed point $P_j$, the other factor is $q+\rho'=2q(\rho(\mu_j))\ne0$, so the zero of $q-\rho'$ has order $\beta_j$. At a branch fixed point $a_j=\rho(\mu_j)$, smoothness gives $\beta_j=1$. Write $F(\rho)=(\rho-a_j)F_j(\rho)$, where $F_j(a_j)=F'(a_j)\ne0$. In the local coordinate $\rho-a_j=t^2$, we may write $y=t s(t^2)$, where $s(\xi)=\sqrt{F_j(a_j+\xi)}$ is a locally chosen holomorphic square root and $s(0)\ne0$. Since $q(a_j)=0$, we have $q=O(t^2)$, so both $q-y$ and $q+y$ have simple zeros in $t$. Thus $q-\rho'$ has order $\beta_j$ in both cases. Their total multiplicity is $M=\sum_{j=1}^{N}\beta_j$. This count is algebraic: with the fixed rows held constant, the determinant is proportional to the meromorphic curve function $q(\rho)-y$. We count its zeros and poles on the compact curve, not in the propagation variable $z$ or on the full Jacobian. Its only pole is at infinity; in a local curve parameter $t$ there, $\rho=t^{-2}$ and $y\sim\pm2t^{-(2g+1)}$, so their pole orders are $2$ and $2g+1$. If $d=\deg q\le g$, the $y$ term gives pole order $2g+1$; if $d\ge g+1$, the leading term of $q$ gives pole order $2d$. A meromorphic function has equal total zero and pole orders, so $M=2g+1$ or $2d$, respectively. The chosen size is $n=g+2$ in the first case and $n=2d-g+1$ in the second. Removing the $n-1$ occurrences represented by the fixed rows therefore leaves $M-(n-1)=g$ further occurrences in either case. These are the remaining intersections of $y=q(\rho)$ with the curve, called the complementary intersection divisor in~\cite[Sections~3.1--3.2]{divisor-addition}. The count identifies the residual points but does not yet prove their sigma factorization. \^Onishi's identity applies with the moving argument constrained to the Abel image of this same curve; below, an auxiliary determinant and the two-point identity establish how the composite factor supplies the residual uncorrected point factors. No global meromorphicity in $z$, or of each individual sigma factor, is required.

In the genus-two example the complete list is $P_1,P_1,P_2,P_3,P_4,P_5$. The repetition of $P_1$ records $\beta_1=2$; in general a point occurs more than once precisely when $\beta_j>1$. The fixed rows of \eqref{eq:onishi-g2-matrix} account for $P_1,P_2,P_3,P_4$, leaving $P_1,P_5$. The determinant's sigma expression must account for these two further point factors, one for mode 1 and one for mode 5. We now identify where they occur and prove their factorization.

Let $B$ index the residual points not yet accounted for by the fixed rows indexed by $A$. We choose the fixed and residual points so each group has distinct coordinates internally, allowing a point to belong to both groups. Thus their combined occurrences reproduce the weights. For the two examples,
\begin{center}
\begin{tabular}{ccc}\toprule
Genus & Fixed points indexed by $A$ & Residual points indexed by $B$\\\midrule
2 & $P_1,P_2,P_3,P_4$ & $P_1,P_5$\\
3 & $P_1,P_2,P_3,P_4,P_5$ & $P_1,P_6,P_7$\\
\bottomrule\end{tabular}
\end{center}
There is no repeated row inside either determinant. The repeated factor will appear when the two lists are combined.

The complete zero list also fixes the summed Abel argument. Abel's theorem says that the weighted sum of the Abel vectors of the zeros of $q-\rho'$, whose poles are all at infinity, is a period vector. For the sigma factorization below, assume compatible fixed lifts have been chosen with this sum zero, and put
\[
\sum_{j=1}^{N}\beta_jw(\mu_j)=0,\qquad
E=\sum_{j\in A}w(\mu_j)=-\sum_{j\in B}w(\mu_j).
\]
Both examples have unit-weight points, so the period vector can be absorbed into one such point's path. Sigma normalizers and $h_{jk}$ must use these same fixed lifts. For the opposite-sheet residual points, use the lifts $-w(\mu_j)$ and the corresponding negative second-kind primitives, giving summed Abel argument $E$. At a branch point the geometric involution fixes the point, but the two lifts differ by a period; retaining the negative lift preserves the sigma quasi-period factors. This is the divisor-class relation used in interpolation-based reduction~\cite{divisor-addition}; here it explains the opposite-sheet points in the auxiliary determinant below.

\^Onishi's identity \cite[Theorem~7.2]{onishi} now groups into
\begin{equation}\label{eq:onishi-composite}
D_n(z)=K_A\frac{\sigma(w(z)+E)}{\sigma_\sharp(w(z))}
                  \prod_{j\in A}\bar v_j(z),
\end{equation}
Here $K_A$ is a nonzero scalar constant formed from the fixed-point sigma values and normalizers in the identity; it is independent of $z$. Their nonvanishing follows from the two-point identity and the distinct fixed coordinates. The $n-1$ factors $\bar v_j$ are the uncorrected point factors of Lemma~\ref{lem:kernel}. The additional ratio $\sigma(w+E)/\sigma_\sharp(w)$ has a summed fixed argument, and is the composite factor we must resolve. The fixed constants use the sign corrected by \^Onishi in his published errata to the determinant identity, $c_n=(-1)^{(n-g-1)(n+g^2+2g)/2}$ for $n\ge g$.

For genus two, apply the sheet involution to the residual points, using $P_1^*=\iota(P_1),P_5^*=\iota(P_5)$ and the moving point in an auxiliary three-row identity. Its columns are purely polynomial, so
\begin{equation}\label{eq:onishi-g2-auxiliary}
\det\begin{pmatrix}
1&\rho(\mu_1)&\rho(\mu_1)^2\\
1&\rho(\mu_5)&\rho(\mu_5)^2\\
1&\rho(z)&\rho(z)^2
\end{pmatrix}
=\Delta_B[\rho(z)-\rho(\mu_1)][\rho(z)-\rho(\mu_5)],
\end{equation}
where $\Delta_B=\rho(\mu_5)-\rho(\mu_1)$. The sheet involution leaves these polynomial entries unchanged, but gives the summed Abel argument $w-w(\mu_1)-w(\mu_5)=w+E$ in its sigma expression. Its other moving factors are $\bar u_1\bar u_5$. The full equality supplied by the sigma identity is
\[
\det\begin{pmatrix}
1&\rho(\mu_1)&\rho(\mu_1)^2\\
1&\rho(\mu_5)&\rho(\mu_5)^2\\
1&\rho(z)&\rho(z)^2
\end{pmatrix}
=K_B\frac{\sigma(w+E)}{\sigma_2(w)}\bar u_1\bar u_5.
\]
Here $K_B$ is another nonzero scalar constant determined by the fixed residual points, independent of $z$.
On the other hand, the two-point identity gives $\bar u_j\bar v_j=\rho-\rho(\mu_j)$ in genus two. Thus \eqref{eq:onishi-g2-auxiliary} is also $\Delta_B\bar u_1\bar v_1\bar u_5\bar v_5$. Comparing these expressions and cancelling $\bar u_1\bar u_5$ proves
\[
\frac{\sigma(w+E)}{\sigma_2(w)}=\Gamma_B\bar v_1\bar v_5,
\qquad \Gamma_B=\Delta_B/K_B.
\]
Substitution into \eqref{eq:onishi-composite} gives $D_5=K\bar v_1^2\bar v_2\bar v_3\bar v_4\bar v_5$, with $K=K_A\Gamma_B$. The composite factor has supplied exactly the missing occurrences: it is not an unexplained extra mode.

In general the auxiliary determinant has the $g$ residual points on the opposite sheet and the moving point, with columns $1,\rho,\ldots,\rho^g$. Its Vandermonde expression contains $\prod_{j\in B}[\rho-\rho(\mu_j)]$; its sigma expression contains $[\sigma(w+E)/\sigma_\sharp(w)]\prod_{j\in B}\bar u_j$. Comparing them through \eqref{eq:kernel-product} proves, by the same cancellation,
\begin{equation}\label{eq:onishi-residual-split}
\frac{\sigma(w+E)}{\sigma_\sharp(w)}
 =\Gamma_B\prod_{j\in B}\bar v_j,
\qquad \Gamma_B\ne0\text{ constant in }z.
\end{equation}
We cancel on a regular chart and extend the resulting identity wherever the continued expressions are defined. This is a factorization along the curve's Abel image, not on the entire Jacobian. Combining the main and residual factors, and then multiplying the two-point identities, yields
\begin{equation}\label{eq:onishi-raw-products}
\prod_{j=1}^{N}\bar u_j^{\beta_j}=K_+[q(\rho)+\rho'],\qquad
\prod_{j=1}^{N}\bar v_j^{\beta_j}=K_-[q(\rho)-\rho'],
\end{equation}
where $K_+,K_-$ are nonzero fixed constants. All the weighted uncorrected point factors have now been recovered from the determinant identity.

\subsection{Step 3: deduce cancellation of the exponential corrections}

The cancellation is a consequence of the complete factorization established in Step 2 together with the differential identities for the modes. The full modes include both the point factors and exponential corrections. Multiplying their normalized sigma expressions gives
\begin{equation}\label{eq:onishi-product-exponential}
\begin{aligned}
U(z)&=U(z_0)\prod_{j=1}^{N}\left[\frac{\bar u_j(z)}{\bar u_j(z_0)}\right]^{\beta_j}\exp\bigl[\mathcal E(z)\bigr],\\
\mathcal E(z)&=\left[y\sum_{j=1}^{N}\beta_jP_j(\rho)\right]_{z_0}^z
+\sum_{k=1}^g\left(\sum_{j=1}^{N}\beta_jd_{jk}\right)\Delta R_k\\
&\quad+\sum_{k=1}^g\left\{\sum_{j=1}^{N}\beta_j(c_{jk}-h_{jk})\right\}\Delta w_k.
\end{aligned}
\end{equation}
The corresponding $V$ formula uses $\bar v_j$ and the negative exponent. We show that each weighted contribution to $\mathcal E$ vanishes.

Step 2 identifies $\prod_{j=1}^{N}\bar u_j^{\beta_j}$ with a constant times $q+\rho'$. Their logarithmic derivatives are equal. Using the kernel differential on one side and the scalar chain rule on the other gives two expressions for this same derivative:
\[
\begin{aligned}
\frac{d}{dz}\log\prod_{j=1}^{N}\bar u_j^{\beta_j}
 &=\sum_{j=1}^{N}\beta_j\frac{\rho'(\mu_j)-\rho'}{2[\rho(\mu_j)-\rho]}
   +\sum_{k=1}^g\left(\sum_{j=1}^{N}\beta_jh_{jk}\right)\rho^{k-1},\\
\frac{d}{dz}\log(q+\rho')
 &=\sum_{j=1}^{N}\beta_j\frac{\rho'(\mu_j)-\rho'}{2[\rho(\mu_j)-\rho]}
   +\sum_{j=1}^{N}\beta_j B_j(\rho).
\end{aligned}
\]
The rational point terms cancel, giving
\[
\sum_{j=1}^{N}\beta_j B_j(\rho)
=\sum_{k=1}^g\left(\sum_{j=1}^{N}\beta_jh_{jk}\right)\rho^{k-1}.
\]
Since $\rho$ is nonconstant, this is a polynomial identity. Multiplication by $d\rho/y$ shows that the weighted correction differential is holomorphic. Uniqueness of the decomposition~\eqref{eq:B-reduction} then gives
\begin{equation}\label{eq:onishi-weighted-gauge}
\begin{aligned}
\sum_{j=1}^{N}\beta_jP_j&=0,\qquad
\sum_{j=1}^{N}\beta_jd_{jk}=0,\\
\sum_{j=1}^{N}\beta_j(c_{jk}-h_{jk})=0&\qquad(k=1,\ldots,g).
\end{aligned}
\end{equation}
Thus the algebraic endpoint terms, second-kind integrals and holomorphic exponential corrections all cancel in the complete products. The exponential in \eqref{eq:onishi-product-exponential} becomes one as a consequence of the complete \^Onishi factorization.

Only the normalized point-factor products remain. Equation~\eqref{eq:onishi-raw-products} gives
\[
U(z)=U(z_0)\frac{q(\rho(z))+\rho'(z)}{q(\rho(z_0))+\rho'(z_0)},\qquad
V(z)=V(z_0)\frac{q(\rho(z))-\rho'(z)}{q(\rho(z_0))-\rho'(z_0)}.
\]
The initial product conditions fix these as $(q+\rho')/2$ and $(q-\rho')/2$. Step 1 then identifies them with the normalized determinants:
\begin{equation}\label{eq:onishi-mode-products}
U(z)=\frac{D_n^*(z)}{\varepsilon_-\Delta_A},\qquad
V(z)=-\frac{D_n(z)}{\varepsilon_+\Delta_A}.
\end{equation}
The constants $\varepsilon_\pm$ are given in \eqref{eq:onishi-longer-signs}, with $\varepsilon_+=\varepsilon_-=1$ for $\deg q\le g$. For the two determinant examples below $\varepsilon_+=\varepsilon_-=-1$, hence $U=-D_n^*/\Delta_A$ and $V=D_n/\Delta_A$. This completes the connection between the Hamiltonian products, the sigma mode solutions and \^Onishi's determinant.

The balanced fixed lifts make the cancellation visible in the explicit exponent. Other compatible lifts can redistribute an exponential between the raw sigma product and its correction, but the normalized initial-value solutions, and therefore their determinant products, are unchanged. The construction uses internally distinct main and residual lists, which both determinant examples satisfy. More general multiplicities could be treated through confluent determinant identities, in which derivative limits replace coincident evaluation rows, but we do not pursue that extension here. Bernatska and Kopeliovich treat rational interpolation and repeated-point constructions in~\cite[Sections~3.3--3.5]{divisor-addition}.

\section{Elementary modes and polynomial second-order potentials}
\label{sec:polynomial-potentials}

Baker--Akhiezer functions and their relation to generalized Lam\'e
equations are discussed by Matsutani~\cite[Section~4.2.2.3]{matsutani-book}.
Here the restriction to the one-point Abel image and the prescribed mode
products lead to second-order equations with self-generating polynomial
potentials.

Write $a_j=\rho(\mu_j)=\gamma_j/\beta_j$ and
$b_j=\rho'(\mu_j)=q(a_j)$, so $I_j=\beta_j(a_j-\rho)$.
The equations in this section follow from the scalar reduction and apply
independently of the leading-four normalisation used for sigma evaluation.

\begin{theorem}[Vanishing correction and constant Wronskian]
\label{thm:special-W}
For modes satisfying~\eqref{eq:split-log} on a nonconstant scalar trajectory,
with correction $B_j$ defined by~\eqref{eq:B-def}, the Wronskian
$W_j=u_jv_j'-v_ju_j'$ satisfies
\begin{equation}\label{eq:special-W}
 W_j=-\beta_j\{b_j+2(a_j-\rho)B_j(\rho)\}.
\end{equation}
It is constant if and only if $B_j=0$ on that power line.
In this case $W_j=-\beta_jb_j$; when $b_j\ne0$, the modes form a
fundamental pair for a common equation $\psi''=G_j(\rho(z))\psi$,
by the converse of Abel's Wronskian identity.
\end{theorem}

\begin{proof}
Multiply the difference of the logarithmic equations~\eqref{eq:split-log}
by $u_jv_j$ to obtain~\eqref{eq:special-W}.
If $W_j$ is constant, the polynomial $(a_j-\rho)B_j(\rho)$ is constant
on the nonconstant trajectory, hence identically constant. Its value at
$\rho=a_j$ is zero, so $B_j=0$. The converse follows immediately.
Finally $W_j'=u_jv_j''-v_ju_j''=0$ implies
$u_j''/u_j=v_j''/v_j$ on a regular chart, and a nonzero Wronskian gives
independence.
\end{proof}

\begin{theorem}[Polynomiality for a prescribed product]
\label{thm:special-potential}
Let a nonconstant analytic trajectory satisfy $(\rho')^2=F(\rho)$,
$\rho''=F_\rho/2$, with $F$ polynomial. Suppose
\[
 uv=\beta(a-\rho),\qquad uv'-vu'=\beta c,
\]
where $\beta\ne0$ and $a,c$ are constant. On a regular chart both modes
solve
\begin{equation}\label{eq:special-potential}
 \psi''=G_{a,c}(\rho)\psi,\qquad
 G_{a,c}(\rho)=\frac14\frac{d}{d\rho}
                    \frac{F(\rho)-c^2}{\rho-a}.
\end{equation}
This potential is polynomial if and only if $c^2=F(a)$.
When polynomial, its degree is $D-2$ for $\deg F=D\ge2$.
The pair is independent if and only if $c\ne0$.
\end{theorem}

\begin{proof}
The prescribed product and Wronskian give
\[
 \frac{u'}u=\frac{-c-\rho'}{2(a-\rho)},\qquad
 \frac{v'}v=\frac{c-\rho'}{2(a-\rho)}.
\]
Differentiation using the scalar equations gives~\eqref{eq:special-potential}.
Polynomial division then isolates its only obstruction:
\begin{equation}\label{eq:special-pole}
 G_{a,c}(\rho)=\frac14\frac{d}{d\rho}
                 \frac{F(\rho)-F(a)}{\rho-a}
              +\frac{c^2-F(a)}{4(\rho-a)^2}.
\end{equation}
The first term is polynomial; the double pole disappears precisely when
$c^2=F(a)$. If $F$ has leading coefficient $f_D$, the polynomial term
has leading coefficient $(D-1)f_D/4$ and degree $D-2$.
Independence follows from the Wronskian.
\end{proof}

For the special canonical modes, $c=-b_j$ and $b_j^2=F(a_j)$, so the
potential is automatically polynomial. Substitution of
$\rho=a_j-I_j/\beta_j$ expresses it as a polynomial in the paired power:
\begin{equation}\label{eq:special-self-generating}
 \begin{aligned}
 \mathcal P_j(I)&=G_{a_j,-b_j}(a_j-I/\beta_j),\\
 u_j''&=\mathcal P_j(u_jv_j)u_j,\qquad
 v_j''=\mathcal P_j(u_jv_j)v_j.
 \end{aligned}
\end{equation}
Its coefficients are fixed by the conserved curve and mode point.
For the prescribed product, the polynomiality condition in
Theorem~\ref{thm:special-potential} fixes the squared normalised Wronskian
$c^2=F(a)$; its two signs interchange $u,v$. The Hamiltonian solutions also satisfy the
complete-product constraints~\eqref{eq:products}.

For completeness, the second-order coefficients of the general modes are
also explicit. Suppressing the index $j$, put $d=a-\rho$, $y=\rho'$ and
$G_b=G_{a,b}$. Differentiating~\eqref{eq:split-log} gives
\begin{equation}\label{eq:special-general-second}
 \begin{aligned}
 u''/u&=G_b+B^2+bB/d+y(B_\rho-B/d),\\
 v''/v&=G_b+B^2+bB/d-y(B_\rho-B/d).
 \end{aligned}
\end{equation}
If $F$ is nonsquare, $1$ and $y$ are independent over $\mathbb C(\rho)$.
Polynomiality in $\rho$ of either coefficient therefore forces
$(a-\rho)B_\rho-B=0$. The polynomial $(a-\rho)B$ is then constant and
zero at $a$, giving $B=0$. Conversely $B=0$ gives the polynomial $G_b$.

For constant $Q=E$, the energy cancels explicitly from the potential:
\begin{equation}\label{eq:special-constant-potential}
 G_j(\rho)=\frac{d}{d\rho}\frac{r(\rho)}{a_j-\rho}.
\end{equation}
Since $\beta_j$ is a positive integer, the quotient is polynomial.
Equivalently, define a polynomial in all the modal powers by
\begin{equation}\label{eq:special-global-potential}
 \mathcal V_j(I_1,\ldots,I_N)
 =-\beta_j\sum_k\beta_k\partial_{I_k}
          \left(I_j^{\beta_j-1}\prod_{\ell\ne j}I_\ell^{\beta_\ell}\right).
\end{equation}
Then $u_j''=\mathcal V_j(I)u_j$ and $v_j''=\mathcal V_j(I)v_j$.
This also follows directly by differentiating~\eqref{eq:canonical} with
$Q_j=0$, so it holds globally as a polynomial identity in the canonical
variables. For unit weights it reduces to
$\mathcal V_j=-\sum_{k\ne j}\prod_{\ell\ne j,k}I_\ell$.
When $E\ne0$, every Wronskian $-\beta_jE$ is nonzero; at $E=0$ each
pair is proportional, though the polynomial second-order equation remains.

Finally, for $F(\rho)=4\rho^3-g_2\rho-g_3$, the potential is
$G_{a,c}=2\rho+a$ when $c^2=F(a)$, recovering the $n=1$ Lam\'e equation.
For higher degrees, the mode point generally enters nonconstant
coefficients as well; the resulting equations need not be a single
shared potential with different additive spectral parameters.

\section{Worked examples}\label{sec:examples}

We first give analytic mode solutions for a five-pair genus-two system and a seven-pair genus-three system with nonzero corrections. We then give three special-polynomial examples illustrating the classification and second-order potentials.

\subsection{Genus-two example with repeated mode}

This example demonstrates the mode-solution Theorem~\ref{thm:mode-solutions}
for a five-pair system on a smooth genus-two curve, including a double
mode and affine polynomial corrections.

\subsubsection{Polynomial data}

Choose degeneracies and power offsets:
\begin{equation}\label{ex:g2-data}
\beta = (2,1,1,1,1), \qquad \gamma = \left(-\frac{3}{5},-\frac{1}{5},-\frac{1}{10},\frac{2}{5},\frac{1}{2}\right).
\end{equation}
Note $\sum_{j=1}^{5}\beta_j = 6$ and $\sum_{j=1}^{5}\gamma_j = 0$. The centred power convention gives
\[
\rho = -\frac{\sum_{j=1}^{5} x_j}{\sum_{j=1}^{5} \beta_j} = -\frac{x_1+x_2+x_3+x_4+x_5}{6}.
\]

Define the polynomial $Q(x_1,\ldots,x_5)$ in independent power variables:
\begin{equation}\label{ex:g2-Q}
\begin{aligned}
Q(x_1,\ldots,x_5)&=4\left(-\frac{S(x)}6+\frac{27}{40}\right)\left(-\frac{S(x)}6+\frac1{20}\right)\left(-\frac{S(x)}6-\frac35\right)\\
&\quad+\frac13\left(\frac{x_1}{2}-x_2+\frac1{10}\right)x_3,
\qquad S(x)=\sum_{k=1}^{5}x_k.
\end{aligned}
\end{equation}

The second term is quadratic and vanishes on the power line, but changes individual partial derivatives. Writing the first term as $q(-S(x)/6)$, the full derivatives are
\[
\begin{aligned}
Q_1&=-q_\rho(-S(x)/6)/6+x_3/6,& Q_2&=-q_\rho(-S(x)/6)/6-x_3/3,\\
Q_3&=-q_\rho(-S(x)/6)/6+(x_1/2-x_2+1/10)/3,& Q_4=Q_5&=-q_\rho(-S(x)/6)/6.
\end{aligned}
\]
One must differentiate the polynomial in the independent powers before restricting to the line: $q$ alone does not specify the individual modes.

\subsubsection{The explicit coupled equations}

For the full system, set $I(z)=(u_1(z)v_1(z),\ldots,u_5(z)v_5(z))$. The polynomial multiplying each $u_j$ is $\Phi_j^{(2)}(I)=-Q_j(I)$; the same polynomial enters the $v_j$ equation with the opposite sign. It is useful to display these terms before restricting to the scalar trajectory. Define
\begin{equation}\label{ex:g2-common-phase}
A_2(x)=\frac{S(x)^2}{18}-\frac{S(x)}{36}-\frac{107}{400},\qquad S(x)=\sum_{k=1}^5x_k.
\end{equation}
Then
\begin{equation}\label{ex:g2-phases}
\begin{aligned}
\Phi_1^{(2)}(x)&=A_2(x)-x_3/6,\\
\Phi_2^{(2)}(x)&=A_2(x)+x_3/3,\\
\Phi_3^{(2)}(x)&=A_2(x)-(x_1/2-x_2+1/10)/3,\\
\Phi_j^{(2)}(x)&=A_2(x),\qquad j=4,5.
\end{aligned}
\end{equation}
The coupled polynomial equations are
\begin{equation}\label{ex:g2-coupled}
\begin{aligned}
u_1'&=\Phi_1^{(2)}(I)u_1+2v_1\prod_{k=2}^5v_k,\\
v_1'&=-\Phi_1^{(2)}(I)v_1-2u_1\prod_{k=2}^5u_k,\\
u_j'&=\Phi_j^{(2)}(I)u_j+v_1^2\prod_{\substack{2\le k\le 5\\k\ne j}}v_k,\\
v_j'&=-\Phi_j^{(2)}(I)v_j-u_1^2\prod_{\substack{2\le k\le 5\\k\ne j}}u_k,
\qquad j=2,\ldots,5.
\end{aligned}
\end{equation}
Every multiplicative term is thus an explicit polynomial in the paired products $u_kv_k$. The separate first equation pair shows the factor $2$ and the extra mode factor required by the double weight. The remaining equations have the unit-weight coupling products.

\subsubsection{Restriction to the power line}

On the power line $x_j = \gamma_j - \beta_j\rho$, we substitute to obtain the restricted polynomials. Define
\begin{equation}\label{ex:g2-restrictions}
Q(\gamma_1-\beta_1\rho,\ldots,\gamma_5-\beta_5\rho) = q(\rho), \qquad
\prod_{j=1}^{5} (\gamma_j-\beta_j\rho)^{\beta_j} = r(\rho).
\end{equation}

Explicitly, the restriction gives
\begin{equation}\label{ex:g2-q}
q(\rho) = 4\left(\rho + \frac{27}{40}\right)\left(\rho + \frac{1}{20}\right)\left(\rho - \frac{3}{5}\right).
\end{equation}

The weighted power product is
\begin{equation}\label{ex:g2-r}
r(\rho)=4\left(\rho + \frac{3}{10}\right)^{2}\left(\rho + \frac{1}{5}\right)\left(\rho + \frac{1}{10}\right)\left(\rho - \frac{2}{5}\right)\left(\rho - \frac{1}{2}\right).
\end{equation}

The scalar polynomial $F(\rho) = q(\rho)^2 - 4r(\rho)$ has degree-six terms that cancel, leaving
\begin{equation}\label{ex:g2-F}
\begin{aligned}F(\rho)={}&4 \rho^{5} - \frac{747 \rho^{4}}{100} - \frac{1581 \rho^{3}}{1000} + \frac{83993 \rho^{2}}{40000}\\&+ \frac{16113 \rho}{100000} + \frac{801}{1000000}.\end{aligned}
\end{equation}

This is squarefree of degree five with leading coefficient $4$, as required for a smooth genus-two curve.

\subsubsection{Polynomial corrections}

The fixed mode points satisfy $\rho(\mu_j) = \gamma_j/\beta_j$, giving
\[
(\rho(\mu_1),\ldots,\rho(\mu_5)) = \left(-\frac{3}{10},-\frac{1}{5},-\frac{1}{10},\frac{2}{5},\frac{1}{2}\right).
\]

Exact evaluation of \eqref{eq:B-def} gives
\begin{align}
B_1(\rho)&=\frac{41 \rho}{60} + \frac{67}{150},\label{ex:g2-B1}\\
B_2(\rho)&=\frac{283}{600} - \frac{\rho}{60},\label{ex:g2-B2}\\
B_3(\rho)&=\frac{7 \rho}{60} + \frac{27}{50},\label{ex:g2-B3}\\
B_4(\rho)&=\frac{23}{200} - \frac{53 \rho}{60},\label{ex:g2-B4}\\
B_5(\rho)&=- \frac{13 \rho}{12} - \frac{9}{100},\label{ex:g2-B5}
\end{align}
All corrections are affine, with unequal linear coefficients. The quadratic terms cancel between the divided difference and the full partial of $Q$, although $Q$ itself is cubic. Thus $\deg B_j=1<g=2$ without a rescaling of the curve. The coefficients $c_{j1},c_{j2}$ in the mode formulas are the constant and linear coefficients displayed above.

\subsubsection{Mode solutions}

The Abel integrals are
\[
w_1(z) = \int_\infty^{(\rho(z),\rho'(z))}\frac{d\rho}{y}, \qquad
w_2(z) = \int_\infty^{(\rho(z),\rho'(z))}\frac{\rho\,d\rho}{y}.
\]

The scalar inversion is
\[
\rho(z) = -\frac{\sigma_1(w(z))}{\sigma_2(w(z))}, \qquad \sigma(w(z)) = 0.
\]

The point factors are given by \eqref{eq:point-factor-sigma} with $g=2$:
\[
\widehat{u}_j(z) = \sqrt{\beta_j}\frac{\sigma(w(z)+w(\mu_j))}{\sigma_2(w(z))\sigma_2(w(\mu_j))}\exp\!\left[-h_{j1}\Delta w_1(z) - h_{j2}\Delta w_2(z)\right],
\]
where $\Delta w_k(z) = w_k(z) - w_k(z_0)$ and the constants $h_{j1}, h_{j2}$ are given by \eqref{eq:h-formulas}.

The complete mode solutions are
\begin{equation}\label{ex:g2-modes}
\begin{aligned}
u_j(z) &= u_j(z_0)\frac{\widehat{u}_j(z)}{\widehat{u}_j(z_0)}\exp\!\left[c_{j1}\Delta w_1(z) + c_{j2}\Delta w_2(z)\right],\\
v_j(z) &= v_j(z_0)\frac{\widehat{v}_j(z)}{\widehat{v}_j(z_0)}\exp\!\left[-c_{j1}\Delta w_1(z) - c_{j2}\Delta w_2(z)\right].
\end{aligned}
\end{equation}

These solve the canonical system for Hamiltonian
\[
H = \prod_{j=1}^{5} u_j^{\beta_j} + \prod_{j=1}^{5} v_j^{\beta_j} - Q(u_1v_1,\ldots,u_5v_5),
\]
with $H=0$ and $I_j-\beta_j I_5=\gamma_j-\beta_j\gamma_5$ for $j<5$.

\subsection{Genus-three example with quadratic corrections}

This example demonstrates the mode-solution Theorem~\ref{thm:mode-solutions}
for a seven-pair system on a smooth genus-three curve with quadratic
polynomial corrections.

\subsubsection{Polynomial data}

Choose degeneracies and power offsets:
\begin{equation}\label{ex:g3-data}
\beta = (2,1,1,1,1,1,1), \qquad \gamma = \left(-\frac{3}{5},-\frac{1}{4},-\frac{1}{5},-\frac{1}{10},\frac{1}{5},\frac{2}{5},\frac{11}{20}\right).
\end{equation}
Note $\sum_{j=1}^{7}\beta_j = 8$ and $\sum_{j=1}^{7}\gamma_j = 0$. The centred power convention gives
\[
\rho = -\frac{\sum_{j=1}^7 x_j}{\sum_{j=1}^{7}\beta_j} = -\frac{x_1+\cdots+x_7}{8}.
\]

Define the polynomial $Q(x_1,\ldots,x_7)$ in independent power variables:
\begin{equation}\label{ex:g3-Q}
\begin{aligned}
Q(x_1,\ldots,x_7)&=4\left(-\frac{S(x)}8+\frac{37}{40}\right)\left(-\frac{S(x)}8+\frac3{20}\right)
\left(-\frac{S(x)}8-\frac3{10}\right)\left(-\frac{S(x)}8-\frac{13}{20}\right)\\
&\quad+\frac17\left(\frac{x_1}{2}-x_2+\frac1{20}\right)x_3^2,
\qquad S(x)=\sum_{k=1}^{7}x_k.
\end{aligned}
\end{equation}

The first term depends only on $\sum_{j=1}^{7} x_j$. The second term vanishes on the power line $x_j = \gamma_j - \beta_j\rho$ but has nonzero partials.

\subsubsection{The explicit coupled equations}

For the full system, set $I(z)=(u_1(z)v_1(z),\ldots,u_7(z)v_7(z))$. The polynomial multiplying each $u_j$ is $\Phi_j^{(3)}(I)=-Q_j(I)$; the same polynomial enters the $v_j$ equation with the opposite sign. It is useful to display these terms before restricting to the scalar trajectory. Define
\begin{equation}\label{ex:g3-common-phase}
A_3(x)=-\frac{S(x)^3}{256}+\frac{3S(x)^2}{1024}+\frac{11S(x)}{128}+\frac{249}{6400},\qquad S(x)=\sum_{k=1}^7x_k.
\end{equation}
Then
\begin{equation}\label{ex:g3-phases}
\begin{aligned}
\Phi_1^{(3)}(x)&=A_3(x)-x_3^2/14,\\
\Phi_2^{(3)}(x)&=A_3(x)+x_3^2/7,\\
\Phi_3^{(3)}(x)&=A_3(x)-2(x_1/2-x_2+1/20)x_3/7,\\
\Phi_j^{(3)}(x)&=A_3(x),\qquad j=4,\ldots,7.
\end{aligned}
\end{equation}
The coupled polynomial equations are
\begin{equation}\label{ex:g3-coupled}
\begin{aligned}
u_1'&=\Phi_1^{(3)}(I)u_1+2v_1\prod_{k=2}^7v_k,\\
v_1'&=-\Phi_1^{(3)}(I)v_1-2u_1\prod_{k=2}^7u_k,\\
u_j'&=\Phi_j^{(3)}(I)u_j+v_1^2\prod_{\substack{2\le k\le 7\\k\ne j}}v_k,\\
v_j'&=-\Phi_j^{(3)}(I)v_j-u_1^2\prod_{\substack{2\le k\le 7\\k\ne j}}u_k,
\qquad j=2,\ldots,7.
\end{aligned}
\end{equation}
Every multiplicative term is thus an explicit polynomial in the paired products $u_kv_k$. The separate first equation pair shows the factor $2$ and the extra mode factor required by the double weight. The remaining equations have the unit-weight coupling products.

\subsubsection{Restriction to the power line}

On the power line $x_j = \gamma_j - \beta_j\rho$, the restriction gives
\begin{equation}\label{ex:g3-q}
q(\rho) = 4\left(\rho + \frac{37}{40}\right)\left(\rho + \frac{3}{20}\right)\left(\rho - \frac{3}{10}\right)\left(\rho - \frac{13}{20}\right).
\end{equation}

The weighted power product is
\begin{equation}\label{ex:g3-r}
\begin{aligned}r(\rho)={}&4\left(\rho + \frac{3}{10}\right)^{2}\left(\rho + \frac{1}{4}\right)\left(\rho + \frac{1}{5}\right)\left(\rho + \frac{1}{10}\right)\\&\times \left(\rho - \frac{1}{5}\right)\left(\rho - \frac{2}{5}\right)\left(\rho - \frac{11}{20}\right).\end{aligned}
\end{equation}

The degree-eight terms cancel in $F=q^2-4r$, giving
\begin{equation}\label{ex:g3-F}
\begin{aligned}F(\rho)={}&4 \rho^{7} - \frac{1539 \rho^{6}}{100} + \frac{74 \rho^{5}}{125} + \frac{160751 \rho^{4}}{20000}\\&- \frac{176721 \rho^{3}}{100000} - \frac{7797647 \rho^{2}}{16000000} + \frac{5803941 \rho}{80000000} + \frac{19247121}{1600000000}.\end{aligned}
\end{equation}
Exact squarefreeness gives a smooth genus-three curve with leading coefficient $4$.

\subsubsection{Polynomial corrections}

Write $B_j(\rho)=c_{j1}+c_{j2}\rho+c_{j3}\rho^2$. The exact coefficients are listed in Table~\ref{tab:g3-gauges}.

\begin{table}[htbp]
\centering
\begingroup
\normalsize
\renewcommand{\arraystretch}{1.4}
\setlength{\tabcolsep}{7pt}
\begin{tabular}{lrrrrrrr}\toprule
$j$ & 1 & 2 & 3 & 4 & 5 & 6 & 7\\\midrule
$c_{j1}$ & $- \dfrac{112129}{224000}$ & $- \dfrac{19673}{44800}$ & $- \dfrac{12343}{32000}$ & $- \dfrac{8151}{32000}$ & $\dfrac{4233}{32000}$ & $\dfrac{8489}{32000}$ & $\dfrac{7397}{32000}$\\[8pt]
$c_{j2}$ & $\dfrac{1551}{2800}$ & $\dfrac{191}{280}$ & $\dfrac{263}{400}$ & $\dfrac{277}{400}$ & $\dfrac{223}{400}$ & $\dfrac{107}{400}$ & $- \dfrac{11}{200}$\\[8pt]
$c_{j3}$ & $\dfrac{261}{560}$ & $\dfrac{65}{112}$ & $\dfrac{27}{80}$ & $\dfrac{11}{80}$ & $- \dfrac{37}{80}$ & $- \dfrac{69}{80}$ & $- \dfrac{93}{80}$\\[4pt]
\bottomrule\end{tabular}
\endgroup
\caption{Genus-three polynomial corrections, generated by restriction of the full partial derivatives of $Q$.}\label{tab:g3-gauges}
\end{table}

The cubic terms cancel, leaving $\deg B_j=2<g=3$, although $Q$ is quartic. All three coordinates contribute through $\sum_{k=1}^3c_{jk}\Delta w_k$. As in genus two, the term vanishing on the power line changes the individual corrections without changing $q$ or the scalar curve.
\FloatBarrier

\subsubsection{Mode solutions}

The Abel integrals are
\[
w_k(z) = \int_\infty^{(\rho(z),\rho'(z))}\frac{\rho^{k-1}d\rho}{y}, \qquad k=1,2,3.
\]

The scalar inversion is
\[
\rho(z) = -\frac{\sigma_{13}(w(z))}{\sigma_{23}(w(z))}, \qquad \sigma(w(z)) = \sigma_3(w(z)) = 0.
\]

The point factors are given by \eqref{eq:point-factor-sigma} with $g=3$:
\[
\widehat{u}_j(z) = \sqrt{\beta_j}\frac{\sigma_3(w(z)+w(\mu_j))}{\sigma_2(w(z))\sigma_2(w(\mu_j))}\exp\!\left[-\sum_{k=1}^3 h_{jk}\Delta w_k(z)\right],
\]
where the constants $h_{j1}, h_{j2}, h_{j3}$ are given by \eqref{eq:h-formulas}.

The complete mode solutions are
\begin{equation}\label{ex:g3-modes}
\begin{aligned}
u_j(z) &= u_j(z_0)\frac{\widehat{u}_j(z)}{\widehat{u}_j(z_0)}\exp\!\left[\sum_{k=1}^3 c_{jk}\Delta w_k(z)\right],\\
v_j(z) &= v_j(z_0)\frac{\widehat{v}_j(z)}{\widehat{v}_j(z_0)}\exp\!\left[-\sum_{k=1}^3 c_{jk}\Delta w_k(z)\right].
\end{aligned}
\end{equation}

These solve the canonical system for the seven-pair Hamiltonian with quadratic mode corrections.

\subsection{Special polynomials and second-order mode equations}\label{sec:special-examples}

These examples demonstrate the special-polynomial classification in
Theorem~\ref{thm:special-Q} and the second-order equations in
Theorems~\ref{thm:special-W}--\ref{thm:special-potential}.
In each case $B_j=0$, so the modes are the corrected point factors
in~\eqref{eq:point-factor-sigma}, normalised by their initial values
as in Theorem~\ref{thm:mode-solutions}, using the curve reparametrization
\eqref{eq:general-leading-normalization} when needed.
Each pair has a nonzero constant Wronskian and solves a common equation
with polynomial potential. Below, $\psi_j$ denotes either $u_j$ or $v_j$,
and primes refer to the Hamiltonian propagation variable $z$.

\paragraph{Constant energy in genus two.}

Take
\begin{equation}
 \boldsymbol\beta=(2,1,1,1),\qquad
 \boldsymbol\gamma=(-3,-1,1,3),\qquad Q=1.
\end{equation}
The fixed coordinates are $(-3/2,-1,1,3)$ and all original ordinates are
$b_j=1$. The scalar polynomial is
\begin{equation}\label{eq:g2-F}
 F(\rho)=16\rho^5-124\rho^3-108\rho^2+108\rho+109.
\end{equation}
It defines a smooth genus-two curve. The Wronskians are $(-2,-1,-1,-1)$.
The first pair solves $\psi_1''=G_1(\rho)\psi_1$, where
\begin{equation}\label{eq:g2-potential}
 G_1(\rho)=16\rho^3-18\rho^2-44\rho+6.
\end{equation}
Since $I_1=u_1v_1=-3-2\rho$, its self-generating form is explicitly
\begin{equation}\label{eq:special-g2-self}
 \psi_1''=-\left(2I_1^3+\frac{45}{2}I_1^2+59I_1+\frac{45}{2}\right)\psi_1.
\end{equation}

\paragraph{Constant energy in genus three.}

Take
\begin{equation}
 \boldsymbol\beta=(2,1,1,1,1,1),\qquad
 \boldsymbol\gamma=(-5,-3,-1,1,3,5),\qquad Q=1.
\end{equation}
The fixed coordinates are $(-5/2,-3,-1,1,3,5)$ and again $b_j=1$.
Here
\begin{equation}\label{eq:g3-F}
 \begin{aligned}
 F(\rho)={}&16\rho^7-460\rho^5-500\rho^4+3144\rho^3\\
          &+5000\rho^2-2700\rho-4499.
 \end{aligned}
\end{equation}
This is a smooth genus-three curve. The first Wronskian is $-2$ and the
others are $-1$. The first pair solves $\psi_1''=G_1(\rho)\psi_1$, with
the degree-five potential
\begin{equation}\label{eq:special-g3-potential}
 G_1(\rho)=24\rho^5-50\rho^4-360\rho^3+300\rho^2+1072\rho-90.
\end{equation}
Its polynomial in the paired power is obtained by substituting
$\rho=-(I_1+5)/2$, since $I_1=u_1v_1=-5-2\rho$.

\paragraph{A nonconstant special polynomial in genus two.}

Take $\boldsymbol\beta=(5,5)$, $\boldsymbol\gamma=(-1,1)$ and
\begin{equation}\label{eq:nonconstant-Q}
 Q(I_1,I_2)=6250g_5^{(-5/2,-5/2)}(I_1/5,I_2/5)+1.
\end{equation}
The total weight is ten. The coefficient $\kappa=6250=2\cdot5^5$
cancels the degree-ten terms in $q^2-4r$. Direct restriction gives
\begin{equation}\label{eq:nonconstant-q}
 q(\rho)=6250\rho^5-625\rho^3+\frac{75}{4}\rho+1,
\end{equation}
and the scalar curve is
\begin{equation}\label{eq:nonconstant-F}
 F(\rho)=12500\rho^5+\frac{3125}{2}\rho^4-1250\rho^3
               -\frac{2375}{16}\rho^2+\frac{75}{2}\rho+5.
\end{equation}
It is squarefree of degree five and hence has genus two. The fixed points
have coordinates $a_1=-1/5$, $a_2=1/5$ and ordinates $b_1=1/4$,
$b_2=7/4$, respectively. Their Wronskians are $-5/4$ and $-35/4$.
Both correction polynomials vanish identically, despite the quintic $Q$.
Both pairs solve $\psi_j''=G_j(\rho)\psi_j$, $j=1,2$, with
\begin{equation}\label{eq:special-nonconstant-potentials}
 \begin{aligned}
 G_1(\rho)&=12500\rho^3-\frac{5625}{8}\rho^2
                         -\frac{2125}{4}\rho+\frac{1025}{64},\\
 G_2(\rho)&=12500\rho^3+\frac{24375}{8}\rho^2
                         -\frac{875}{4}\rho-\frac{3775}{64}.
 \end{aligned}
\end{equation}
Substituting $\rho=a_j-I_j/5$ expresses each cubic in its paired power
$I_j=u_jv_j$. The individual-mode reconstruction applies through~\eqref{eq:general-leading-normalization}.
For $\deg q=5$ and $g=2$, Section~\ref{sec:longer-columns} prescribes a nine-row determinant, but these two distinct coordinates cannot supply its eight distinct fixed rows. This example would require confluent determinants to represent the complete products.

\section{Numerical validation}\label{sec:numerical}

We compare the sigma solutions with independent numerical integration of the full polynomial canonical equations, first on an exterior genus-two branch and then over three oscillations of a bounded genus-three branch. Numerous symbolic and numerical checks were carried out in Python; the code is provided in the repository identified in the code availability statement.

Both comparisons use 30-digit working precision. This specifies the arithmetic, not the achieved accuracy: the observed discrepancies and their reduction under RK4 step halving provide the accuracy evidence. Let $s$ denote the state vector containing all $u_j$ and $v_j$ modes. The scaled mode discrepancy at an observation time is
\[
E_{\rm mode}=\max_\ell\frac{|s_\ell^{\rm RK4}-s_\ell^\sigma|}{\max(1,|s_\ell^\sigma|)}.
\]
The sigma solution is sampled on a sparse grid, while RK4 uses a finer grid containing those observation times. Step halving distinguishes finite-step integration error from arithmetic precision. At each observation time, numerical inversion of the scalar Abelian integral determines $\rho(z)$. We then evaluate the Abel vector and the prescribed sigma formulas for the individual modes; the scalar sigma ratio provides a consistency check. The polynomial and exponential coefficients are prescribed analytically, without fitting to trajectory data.

Table~\ref{tab:analytic-checks} reports the existing sigma-based consistency checks separately from the RK4 comparisons. The scalar sigma ratio is compared with the quadrature-inverted coordinate; paired powers and complete products are compared with their prescribed curve expressions. Coordinate discrepancies are divided by $\max(1,|\rho|)$, paired-power discrepancies by $\max(1,|I_j^\sigma|,|\gamma_j-\beta_j\rho|)$, and complete-product discrepancies by $\max(1,|U^\sigma|)$ or $\max(1,|V^\sigma|)$. Here the superscript $\sigma$ denotes evaluation from the sigma mode formulas. The genus-three coordinate lies below one in magnitude, so its absolute and scaled coordinate discrepancies coincide. These checks test the coordinate representation and algebraic identities at the computed Abel vectors; they do not replace scalar inversion at prescribed propagation or provide certified bounds on all sigma-evaluation errors.

\begin{table}[htbp]
\centering\small
\begin{tabular}{lcc}\toprule
Maximum scaled discrepancy & Genus two & Genus three\\\midrule
Scalar sigma coordinate & $5.672682\times10^{-31}$ & $9.423342\times10^{-30}$\\
Paired modal powers & $8.376260\times10^{-31}$ & $5.994113\times10^{-29}$\\
Complete product $U$ & $1.461370\times10^{-30}$ & $8.992481\times10^{-30}$\\
Complete product $V$ & $2.264083\times10^{-30}$ & $2.818472\times10^{-29}$\\
\bottomrule\end{tabular}
\caption{Sigma-based consistency checks at 30-digit working precision, over 33 genus-two observations and 97 uniform genus-three observations plus 18 turning-point and modal-zero observations. Values are maxima from the supplied numerical reports.}\label{tab:analytic-checks}
\end{table}

\subsection{Extended genus-two trajectory}

For the first example, we propagate from $\rho(z_0)=4.5$ to $\rho=10$ on the positive sheet. The propagation length is $T=0.02960302978454372$. We use the Hamiltonian, weights and invariant offsets of the first worked example, with initial modes in the proportions
\[
(u_1(z_0),\ldots,u_5(z_0))\ \propto\
(1.1+0.2i,\;0.8-0.1i,\;1.2+0.1i,\;1+0.1i,\;1.1+0.1i).
\]
Their common factor is the principal sixth root that enforces
\[
u_1(z_0)^2\prod_{j=2}^5u_j(z_0)=\frac{q(4.5)+\sqrt{F(4.5)}}2,\qquad
v_j(z_0)=\frac{\gamma_j-4.5\beta_j}{u_j(z_0)}.
\]
This specifies the initial state on the prescribed invariant level.

The largest real branch point is approximately $1.92552416$. Our exterior branch therefore has $\rho'(z)=\sqrt{F(\rho(z))}>0$ throughout this interval, and its powers are monotone rather than oscillatory. Since $F(\rho)\sim4\rho^5$, the remaining scalar propagation integral to infinity is finite. We stop at $\rho=10$ within a regular chart; the figures do not test passage through a turning point or continuation through infinity. The complex initial modes do not impose $v_j=\overline{u_j}$.

The sigma formulas are evaluated at 33 equally spaced propagation values. Independent RK4 integrations of the full ten-component polynomial field use 32,768, 65,536 and 131,072 uniform steps over the same interval. Every comparison time is a node of each grid, so no interpolation is used. The finest integration supplies 513 saved states for the plotted lines. Table~\ref{tab:g2-extended} reports the maximum of the scaled mode error over all ten components and all 33 observations, rather than only the endpoint.

\begin{table}[htbp]
\centering\small
\begin{tabular}{rccc}\toprule
Steps & Step size & Maximum mode error & Reduction factor\\\midrule
32,768 & $9.0341277\times10^{-7}$ & $2.1807079\times10^{-20}$ & ---\\
65,536 & $4.5170639\times10^{-7}$ & $1.3630311\times10^{-21}$ & $15.998960$\\
131,072 & $2.2585319\times10^{-7}$ & $8.5192239\times10^{-23}$ & $15.999475$\\
\bottomrule\end{tabular}
\caption{Extended genus-two trajectory at 30-digit working precision. Step halving gives the expected fourth-order reduction in RK4 error.}\label{tab:g2-extended}
\end{table}

For the finest grid, the maximum scaled Hamiltonian and weighted-invariant drifts at the observation times are $7.6908187\times10^{-24}$ and $1.2009690\times10^{-22}$. The Hamiltonian drift is divided by $\max(1,|U|,|V|,|Q|)$ at each observation; the drift of $C_j=I_j-\beta_jI_5$ is divided by $\max(1,|\gamma_j-\beta_j\gamma_5|)$. These results distinguish arithmetic precision from finite-step integration accuracy.

The scalar coordinate at each observation is obtained by inversion of its signed quadrature, and the modes are then evaluated from their sigma expressions.

\begin{figure}[!htbp]
\centering\includegraphics[width=\linewidth,height=0.75\textheight,keepaspectratio]{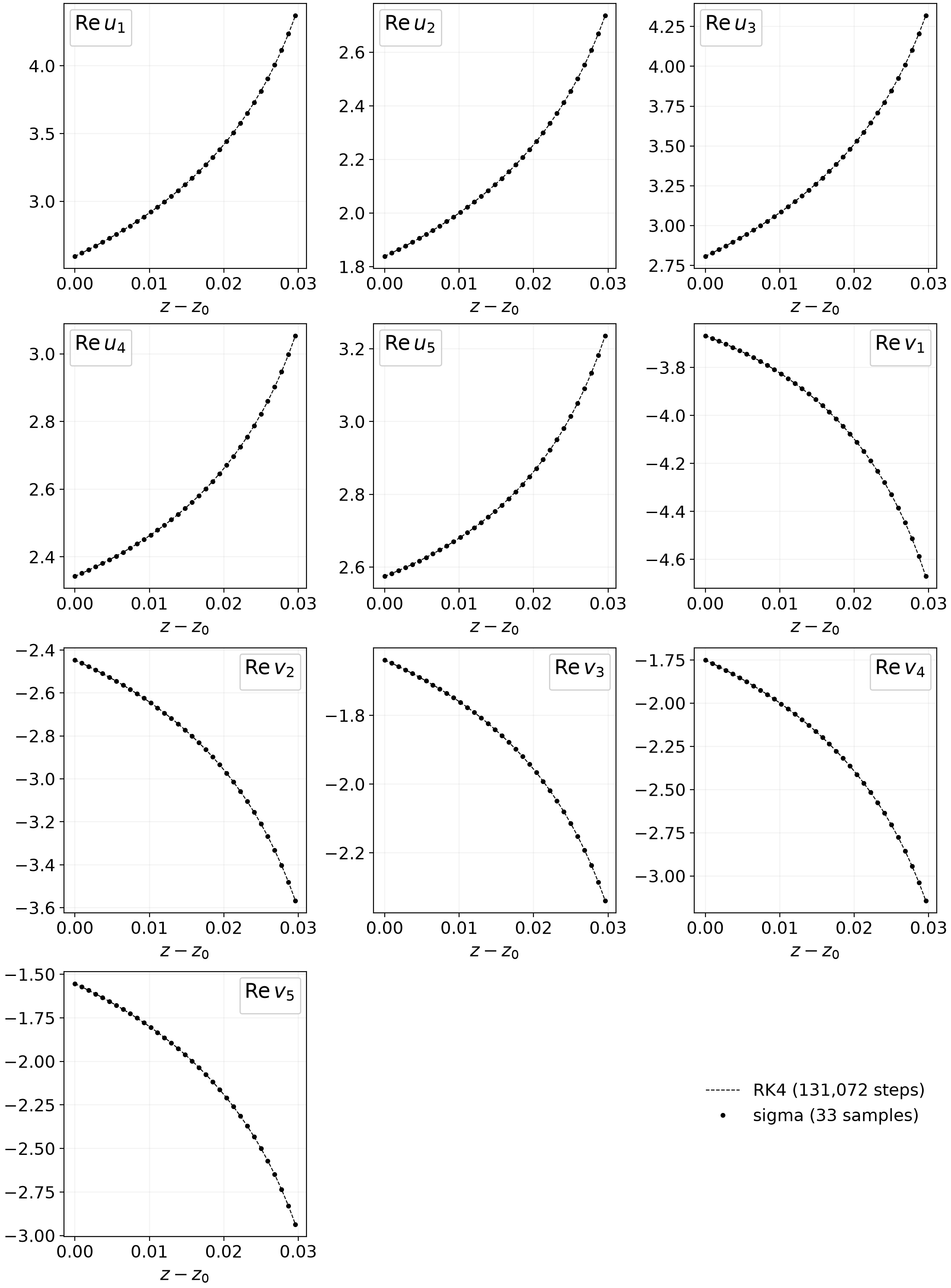}
\caption{Real parts of all five genus-two $u_j$ and $v_j$ modes, ordered across rows with the $u_j$ group first, on the positive-sheet trajectory $4.5\le\rho\le10$, $0\le z-z_0\le0.02960303$. Thin dashed black lines show full-system RK4 with 131,072 steps; black markers show 33 evaluations of the respective sigma mode formulas. Actual amplitudes are shown without individual rescaling.}\label{fig:g2-modes}\label{fig:g2-v-modes}
\end{figure}

\FloatBarrier

\subsection{Oscillatory genus-three trajectory}\label{sec:g3-oscillatory}

The genus-three example admits bounded scalar motion on the same Hamiltonian and invariant level. All seven roots of \eqref{ex:g3-F} are real and simple. Between the fifth and sixth roots,
\[
L\approx0.3505953510803235,\qquad R\approx0.5739427200213239,
\]
we have $F(\rho)\ge0$. The scalar coordinate oscillates between these turning points, so the powers $I_j=\gamma_j-\beta_j\rho$ are periodic. Their period is
\[
T=2\int_L^R\frac{d\rho}{\sqrt{F(\rho)}}\approx6.624556082295094.
\]
We follow three cycles, over $0\le z-z_0\le3T\approx19.87366824688528$.

We start at $\rho(z_0)=(L+R)/2$ with $\rho'(z_0)=+\sqrt{F(\rho(z_0))}$. The initial $u_j$ are positive real numbers in the proportions
\[
(u_1(z_0),\ldots,u_7(z_0))\ \propto\ (1.1,\;0.8,\;1.2,\;1,\;1.1,\;0.9,\;1.05).
\]
Their common scale is determined by the initial product condition
\[
u_1(z_0)^2\prod_{j=2}^7u_j(z_0)
 =\frac{q(\rho(z_0))+\rho'(z_0)}2,\qquad
v_j(z_0)=\frac{\gamma_j-\beta_j\rho(z_0)}{u_j(z_0)}.
\]
The complete modes are then given by \eqref{ex:g3-modes}. The sixth and seventh powers cross zero at $\rho=0.4$ and $0.55$; the polynomial coupled equations remain regular there. These initial conditions do not impose $v_j=\overline{u_j}$, so the signed powers are not positive optical intensities.

For continuation, we follow the signed curve point $(\rho(z),\rho'(z))$ through each turning point and integrate the Abel differentials along that path, keeping $\Delta w_1=z-z_0$. A regular curve parameter removes the square-root singularity of the $\rho$ coordinate at a turning point. The Abel vector retains its accumulated periods as the point completes successive cycles. Both $u_j$ and $v_j$ are evaluated from their own sigma expressions, which allow mode zeros without division by a vanishing mode.

Figures~\ref{fig:g3-scalar}--\ref{fig:g3-v-modes} show the scalar coordinate, periodic powers and individual modes. The modes can have growing or decaying envelopes even when their products repeat: the sigma expressions acquire reciprocal cycle multipliers for $u_j$ and $v_j$. After three scalar cycles, for example, $u_6(z_0+3T)/u_6(z_0)\approx-1.0997562$ and $u_7(z_0+3T)/u_7(z_0)\approx-0.025748382$. The plots retain these envelopes and sign reversals.

Independent full-system RK4 integrations are compared with 97 equally spaced sigma evaluations on common grid nodes. Table~\ref{tab:g3-oscillatory} reports the maximum scaled mode discrepancy over all fourteen components and all 97 observations.

\begin{table}[htbp]
\centering\small
\begin{tabular}{rccc}\toprule
Steps & Step size & Maximum mode error & Reduction factor\\\midrule
24,576 & $8.0866163\times10^{-4}$ & $2.5450058\times10^{-13}$ & ---\\
49,152 & $4.0433082\times10^{-4}$ & $1.5925384\times10^{-14}$ & $15.980813$\\
98,304 & $2.0216541\times10^{-4}$ & $9.9593318\times10^{-16}$ & $15.990414$\\
\bottomrule\end{tabular}
\caption{Full-system RK4 refinement over three genus-three scalar cycles. The same 97 sigma observations are used on all grids, at 30-digit working precision.}\label{tab:g3-oscillatory}
\end{table}

The step-halving factors are close to 16, as expected for fourth-order RK4. On the finest grid, the maximum scaled Hamiltonian and weighted-invariant drifts are $1.5017527\times10^{-17}$ and $3.6848896\times10^{-18}$, with the same scaling conventions as above and $C_j=I_j-\beta_jI_7$. The mode discrepancy of approximately $10^{-15}$ measures finite-step integration error over three cycles, separately from the precision of the sigma evaluations.

\begin{figure}[p]
\centering\includegraphics[width=\linewidth]{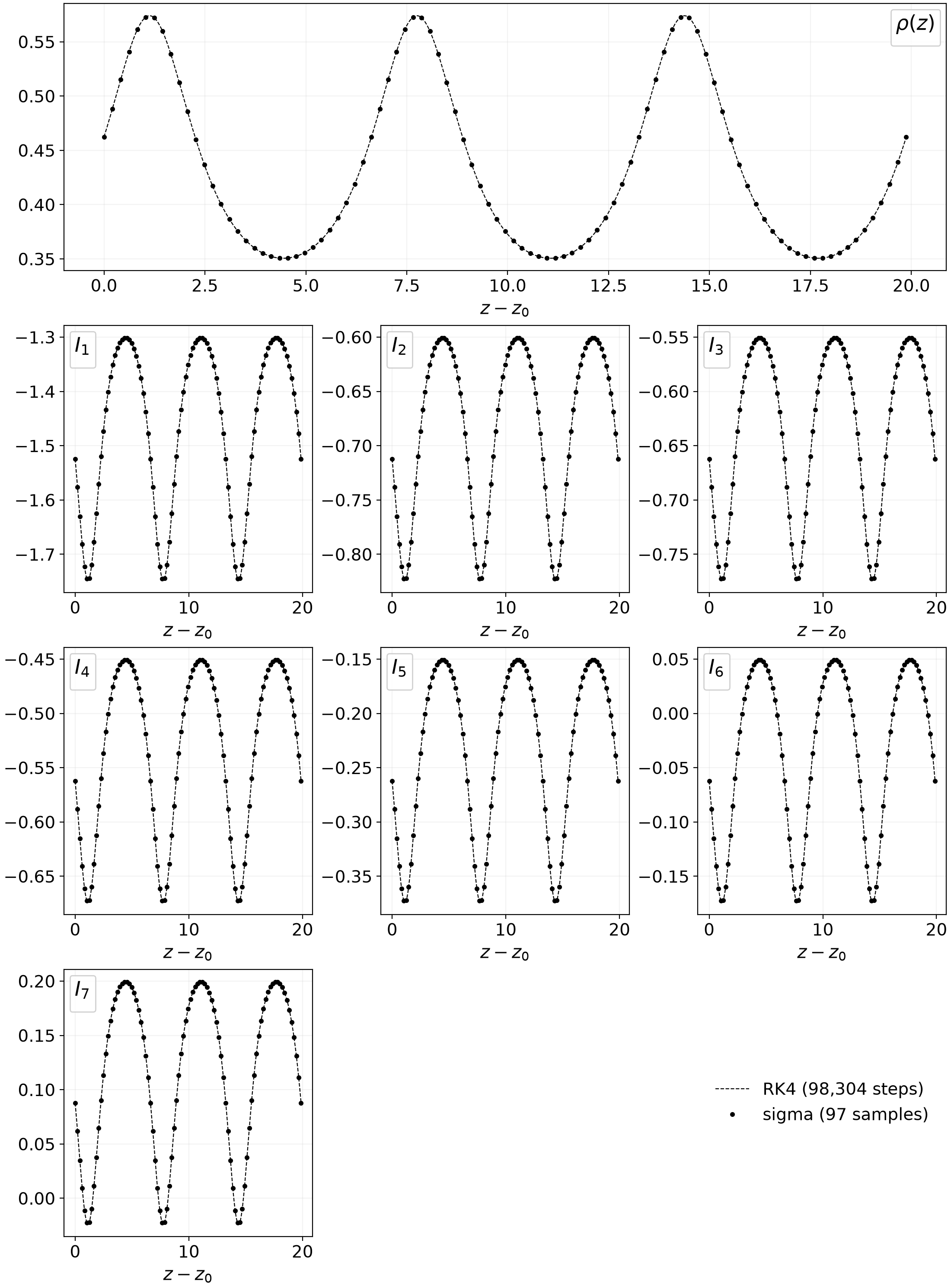}
\caption{Genus-three scalar coordinate $\rho$ (top) and all seven signed powers $I_j$ over three cycles, $0\le z-z_0\le19.87366825$. Their common oscillations show the affine relations $I_j=\gamma_j-\beta_j\rho$; the sixth and seventh powers cross zero. Thin dashed black lines show full-system RK4 with 98,304 steps; black markers show 97 sigma observations, using the coordinate ratio for $\rho$ and the mode products for $I_j$.}\label{fig:g3-scalar}\label{fig:g3-powers}
\end{figure}
\begin{figure}[p]
\centering\includegraphics[width=\linewidth]{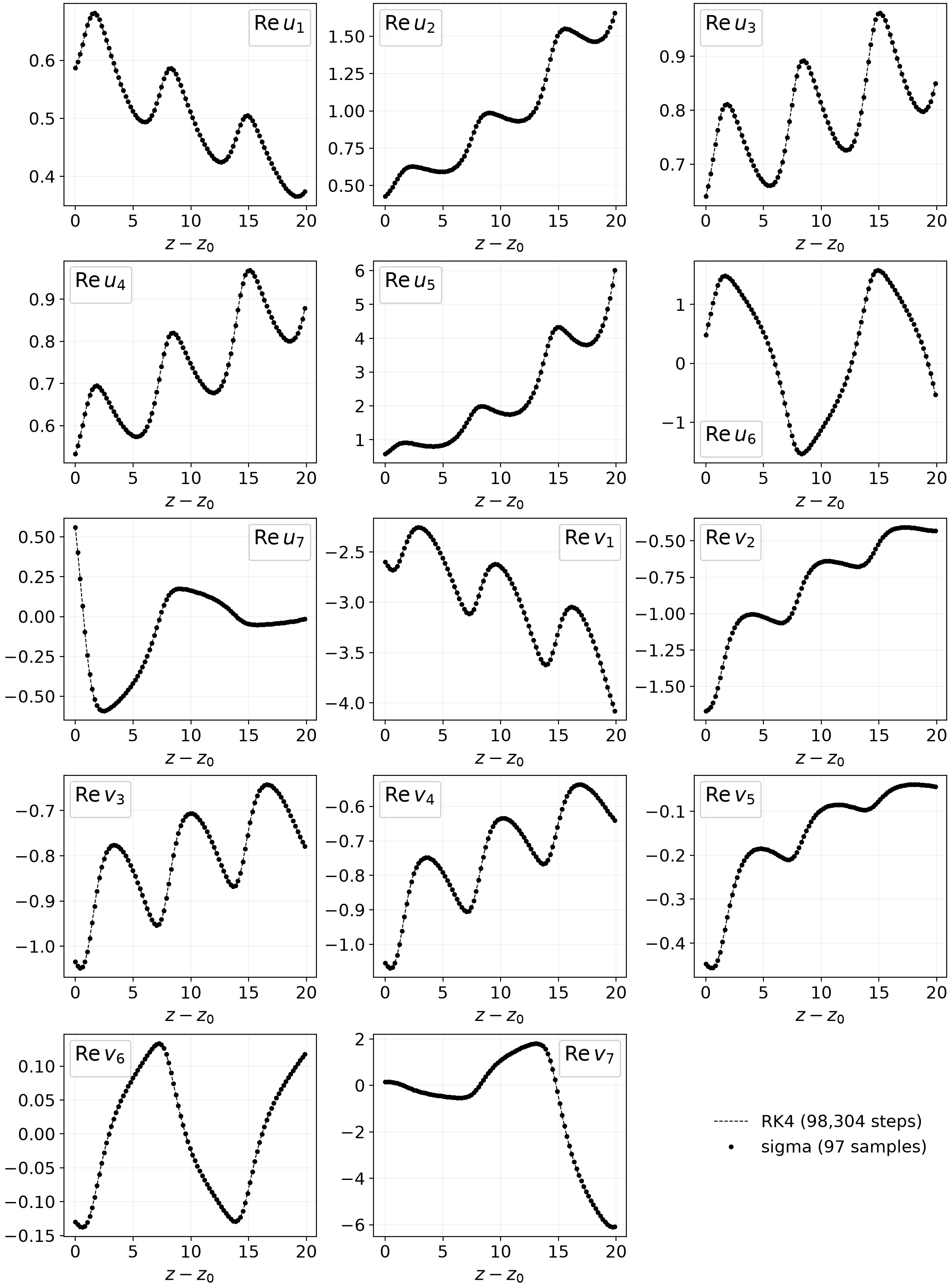}
\caption{Real parts of all seven genus-three $u_j$ and $v_j$ modes over the same three scalar cycles as Figure~\ref{fig:g3-scalar}, ordered across rows with the $u_j$ group first. Thin dashed black lines show full-system RK4; black markers show direct sigma mode evaluations. Actual amplitudes retain their cycle multipliers, envelopes and sign reversals.}\label{fig:g3-modes}\label{fig:g3-v-modes}
\end{figure}

\FloatBarrier

\subsection{Special-polynomial examples}

The three examples in Section~\ref{sec:special-examples} were checked
by exact algebra, independent integration of the Hamiltonian and
second-order equations, and sigma reconstruction of both modes with
analytic first and second derivatives. The verification records 154 exact
algebra checks, 30 exact normalisation checks and 680 numerical sigma
checks. At 35-digit working precision on the tested short regular paths,
the maximum scaled sigma discrepancies are respectively
$1.25\times10^{-33}$, $1.75\times10^{-31}$ and $9.27\times10^{-33}$.

Following~\eqref{eq:general-leading-normalization}, the curve ordinate is rescaled for sigma evaluation so that the
coefficient of the highest power of $\rho$ is four. For the first two
examples, use $\tau=2z$ and $Y=\rho'/2$, giving $Y^2=F/4$.
For the third, use $\tau=25\sqrt5\,z$ and $Y=\rho'/(25\sqrt5)$,
giving $Y^2=F/3125$. The fixed ordinates are divided by the same factors,
and derivatives are converted back to $z$. These changes apply to the
curve evaluation; the modes, Wronskians and second-order equations
displayed in Section~\ref{sec:special-examples} retain the original
Hamiltonian normalisation.

\section{Discussion}\label{sec:discussion}

The explicit solutions raise two further questions: what the conserved quantities establish about integrability, and how the local formulas behave beyond a chosen chart. We separate the general analytic construction from the scope of the examples and identify the remaining continuation and spectral problems.

\subsection{Scope of the construction}

The construction combines scalar reduction with individual-mode integration for the polynomial Hamiltonian class considered here. It requires a smooth odd-degree curve, distinct weighted fixed mode points and a regular mode chart; curve evaluation uses the leading-four normalization~\eqref{eq:general-leading-normalization}. Polynomial corrections integrate through algebraic endpoint terms and first- and second-kind differentials. For $\deg B_j\le g-1$ only first-kind differentials are needed, as in the first two examples. The determinant column list is determined by the scalar degree; the sigma product factorization uses internally distinct main and residual lists with compatible balanced lifts. The determinant construction provides an additional interpretation of the complete mode products; its distinct-point partition assumptions do not restrict Theorem~\ref{thm:mode-solutions}. Universally special polynomials are classified in Theorem~\ref{thm:special-Q}; their correction is identically zero. When $b_j\ne0$, each elementary mode pair forms a fundamental pair for a second-order equation with a self-generating polynomial potential, as established in Section~\ref{sec:polynomial-potentials}. This recovers the $n=1$ Lam\'e equation in genus one and gives higher-degree polynomial potentials in the higher-genus examples. Integer weights, including a double-weight mode, enter through the power line and weighted products.

Regularized second-kind integrals define the exponential coefficients in arbitrary genus. The short fixed-point sigma expressions and full canonical trajectory checks in this paper cover genera two and three. Scalar coordinate recovery is given by the arbitrary-genus derivative ratio, with an additional genus-four formula and a separate check at regular curve points. Finite sigma-based coefficient extraction is available in every genus; uniform short fixed-point expressions remain a separate question. Singular curves and even-degree models require further treatment.

For genus one the same structural split gives a shifted Weierstrass sigma ratio with a zeta exponential. Its sole Abel coordinate still satisfies $w_1(z)=w_1(z_0)+z-z_0$. In higher genus only $w_1$ evolves linearly, the one-point image lies on the theta divisor, and sharp/flat derivatives provide the mode kernel. The scalar inversion ratios are separate restrictions to that image.

At zero energy, the branch-point mode formulas also provide a route to
generalized Jacobi functions associated with extended Euler systems,
including those studied by Ferrer, Crespo and Molero~\cite{ferrer-ees}.
Their interpretation as normalized one-point restrictions of hyperelliptic
$\mathrm{al}$-functions, together with their polynomial second-order
equations, will be explored separately.

\subsection{Commuting integrals and local integrability}

The weighted powers in \eqref{eq:weighted-invariants}, together with $H$, give $N$ commuting integrals. With the canonical bracket
\[
\{A,B\}=\sum_{j=1}^{N}\left(\partial_{u_j}A\,\partial_{v_j}B
-\partial_{v_j}A\,\partial_{u_j}B\right),
\]
the powers commute mutually and satisfy $\{I_j,H\}=\beta_j(V-U)$. Hence $\{C_j,H\}=\{C_j,C_k\}=0$.

On the open set where every $I_j\ne0$, the differentials $dI_j$ are independent and the $N-1$ weighted combinations $dC_j$ have rank $N-1$. The tangent variation $u_1\mapsto e^t u_1$, $v_1\mapsto e^{-t}v_1$ preserves every power but gives $dH/dt=\beta_1(U-V)$. Therefore $dH$ is independent of the $dC_j$ where $U\ne V$. This proves complex local Liouville integrability on that open set; algebraic complete integrability is not established here. The short regular trajectory charts satisfy this condition. The oscillatory genus-three path also reaches turning points and modal zeros outside that open set; the independence argument above applies to its regular segments. A real compact-torus interpretation would additionally require reality conditions, regularity and compactness of the real invariant levels.

\subsection{Continuation and spectral questions}

The elementary Baker--Akhiezer restrictions and the complete modes are evaluated locally on chosen Abel lifts. Exponentials of Abel integrals and sigma quasi-periods can produce multipliers under continuation; global meromorphicity on the compact curve is not asserted. Determining their monodromy, period structure, and a corresponding spectral problem or Lax representation remains open. The scalar system is regular at finite turning points. The three-cycle calculation in Section~\ref{sec:g3-oscillatory} illustrates the resulting mode multipliers on one real oval; their structure on all cycles of the curve remains to be determined.

For repeated weights, the number of coordinate pairs $N$ differs from the total weight $\sum_{j=1}^{N}\beta_j=\deg r$; the examples exploit that distinction and polynomial cancellation. It would be useful to characterize the most economical mode counts compatible with smooth curves and the distinct-coordinate determinant partition, and to extend the sigma product argument to confluent determinants.

Further questions concern short sigma-only formulas for the higher-genus second-kind primitives, even-degree curve models, and conditions on $Q,\gamma_j,\beta_j$ compatible with real positive powers and physical conjugacy. The explicit low-genus examples provide a starting point for those questions while retaining a precise local domain for the proved formulas.

\section*{Code availability}

The project repository, \url{https://github.com/HeskethGD/hyp-ham-sys}, contains the computational notebooks, numerical verification scripts and figure-generation routines used in this work, together with the manuscript source, executed notebook exports and saved numerical results.
The software is distributed under the MIT license, with its scope specified in the repository; that license does not apply to the manuscript or its figures. Dependency versions and execution instructions are recorded in the repository.

\end{document}